\documentclass[runningheads,envcountsame,envcountsect]{llncs}

\newif\ifanon
\anonfalse

\usepackage{amsmath}
\usepackage{amssymb}
\usepackage{graphicx}
\usepackage{textcomp}
\usepackage{array}
\usepackage[section]{placeins}
\usepackage[nospace]{cite}
\usepackage{hyperref}

\newcommand{\mypara}[1]{\smallskip\noindent\textbf{#1.}}
\newcommand{\temph}[1]{\textbf{#1}}

\ifanon
\newcommand{\repo}{\url{https://anonymous.4open.science/r/GLP-02A3}}
\else
\newcommand{\repo}{\url{https://github.com/EShapiro2/GLP}}
\fi

\newcommand{\calN}{\mathbb{N}}

\newcommand{\ia}{\textit{i}}
\newcommand{\ib}{\textit{ii}}
\newcommand{\ic}{\textit{iii}}

\usepackage{enumitem}
\setlist{nosep, leftmargin=*}
\setlist{itemsep=1pt, topsep=3pt, leftmargin=*}

\title{Grassroots Bonds:\\ Financing by the People, for the People}
\titlerunning{Grassroots Bonds}
\ifanon
\author{Anonymous Submission}
\authorrunning{Anonymous Submission}
\institute{}
\else
\author{Ehud Shapiro}
\authorrunning{E. Shapiro}
\institute{London School of Economics, UK, and Weizmann Institute of Science, Israel}
\fi

\begin{document}

\maketitle

\begin{abstract}
Grassroots currencies turn mutual trust into liquidity: a grassroots coin is a unit of its issuer's debt, backed by the issuer's goods and services, which the issuer must redeem, 1-for-1, against any coin they hold; liquidity arises from mutual credit lines, formed by the voluntary exchange of coins among persons who trust each other.  As coins are redeemable 1-for-1, the exchange must be 1-for-1 as well, lest prompt redemption after it leave one party with undue profit.  Thus, grassroots coins are incongruent with interest-bearing credit.

Here, we extend grassroots currencies to include also grassroots bonds,  units of their issuer's debt due at a later date.  Upon maturity, the bearer of the bond may redeem it against a coin of the issuer, so liquid coins can be lent against interest-bearing bonds.  We specify bonds by adding Date and Escrow clauses to the social contract of grassroots currencies: the contract enforces the redemption of a bond when it is mature according to the date stated by the issuer, who undertakes to keep it current.  

We show that the voluntary swap of coins and bonds can realise the basic financial instruments: loans, sale of debt, and forward contracts, and, with an escrow agent, loans with payment schedules, options, collateral, guarantees, insurance, credit default swaps, letters of credit, and credit lines.  We extend the liquidity ratios of corporate finance to bonds, with maturity as asset class, and prove that a community clears its debts by redemptions exactly when no member owes more than they hold, and that it does so without coordination once dates advance and persons act on their rights.  

Grassroots currencies that include coins and bonds are implemented in GLP, a concurrent logic programming language running on Dart, as a program derived from the contract and demonstrated by a village market of six agents and an escrow agent.
\end{abstract}

\ifanon\else
\medskip\noindent
\textbf{Note.}  Grassroots coins, and the contract, results, and implementation of grassroots currencies that this paper extends, are the subject of a separate paper~\cite{shapiro2024gc}.
\fi

\section{Introduction}\label{sec:introduction}

\mypara{Grassroots currencies}
Grassroots coins~\cite{shapiro2024gc} are units of debt that can be issued and traded digitally by any person---natural or legal---including people, communities, cooperatives, corporations, banks, municipalities and governments.  A coin is backed by the goods and services of its issuer, priced in their own coins, and by the coins they hold: upon request, a person must redeem any coin they have issued, 1-for-1, against any coin they hold.  Liquidity arises from mutual credit lines, formed by the voluntary exchange of coins among persons who know and trust each other, and coin redemption pegs mutually-liquid currencies at par and carries payments along chains of holdings.  Grassroots currencies aim to provide a foundation for grassroots digital economies, allowing them to emerge without initial capital or external credit, operating solely on the networked smartphones of their members~\cite{shapiro2024gc}.

\mypara{Coins are not interest-friendly}
Since grassroots coins can be redeemed 1-for-1, mutual credit must also be 1-for-1, lest closing it right after formation would leave the parties with undue profit or loss. Concretely, if $p$ opens a mutual line of credit with $q$ with an upside for $p$, e.g.\ 100 $p$-coins against 110 $q$-coins, then coin redemption would allow $p$ to immediately redeem the 100 $q$-coins they hold, leaving $p$ with 10 $q$-coins and $q$ with nothing; this is as if a bank could grant a loan with an upfront interest payment then immediately recall the loan while keeping the interest.  More generally, 1-for-1 redemption precludes associating interest with mutual credit lines denominated in grassroots currencies.

\mypara{Grassroots bonds}
In this paper we extend grassroots currencies to include also grassroots bonds, units of their issuer's debt due at a later date: upon maturity, the bearer of a bond may redeem it against a coin of the issuer, so liquid coins can be lent against interest-bearing bonds.  Redemption generalises the coin-for-coin redemption of grassroots currencies, the bearer of a coin taking any coin or bond its issuer holds, so liquid coins are lent in exchange for interest-bearing bonds maturing later, enabling incentive-based liquidity.  We show that the voluntary swap of coins and bonds can realise the basic financial instruments: loans, sale of debt, and forward contracts (Section~\ref{sec:instruments}), and, with an escrow agent, loans with payment schedules, options, collateral, guarantees, insurance, credit default swaps, letters of credit, and credit lines (Section~\ref{sec:escrow} and Appendix~\ref{app:escrow-more}); and that the liquidity ratios of corporate finance apply to grassroots bonds, with maturity as asset class (Section~\ref{sec:liquidity}).  The instruments given are representative, each specified as a transaction of the model, and further instruments are combinations of them.  Those of Section~\ref{sec:escrow} employ an escrow agent, trusted by the parties for that agreement, and, where the outcome turns on an outside event, an oracle attesting it; each construction states the trust it requires.

\mypara{Grassroots bonds with pen and paper}
Grassroots currencies were introduced as cryptocurrencies and reintroduced as a physical currency, realised with pen and paper, a thought experiment that simplifies their exposition and uncovers their inherent properties as a currency~\cite{shapiro2024gc}; we follow it.  A \emph{paper bond} is a paper coin with a maturity date written on it: a promissory note stating its issuer, its maturity date, and three standing obligations of its issuer---to price their offerings in their own coins and accept them in payment at these prices; to price every coin and bond they hold, 1-for-1, against their own coins, on standing offer, with redemption the bearer exercising the offer; and to give one of their own coins for the bond once the bond has matured.  Each person keeps a calendar, and a bond matures by the calendar of its issuer; a bond is transferable at any time, is not paid, and once it is mature its bearer may redeem it against a coin of its issuer (Section~\ref{sec:bonds}).

\mypara{Specification by a grassroots social contract, and the attested date}
We specify grassroots bonds by a \emph{grassroots social contract}~\cite{shapiro2026formalising} that extends the contract of grassroots currencies~\cite{shapiro2024gc} with two clauses: Date, by which a party keeps their date at the current date and warrants it whenever a matured bond of their own is taken for one of their coins; and Escrow, by which a party deposits coins and bonds with a second party for a third (Section~\ref{sec:contract}).  The clauses that describe acts are written as act schemas, from which the volition-guarded multiagent atomic transactions~\cite{lewis2026volitional} realising the contract are compiled, and each clause is enforced, attested, or undertaken by the contract~\cite{shapiro2026formalising}.  The four acts of grassroots currencies are enforced---the runs determine whether their clauses are kept and every correct run keeps them---and the pricing undertaking is witnessed by every coin in another's hands.  The Date clause is attested: a party whose date lags behind the calendar breaches it, and every maturing of one of their bonds records at the bearer the date they warranted.  Whether a bond is mature is subject to its issuer advancing their date.  So the contract enforces the maturing of a bond only against the date its issuer warrants, and the obligation of a bond as a whole---to give a coin for it once its maturity date has come---is attested: its breach is witnessed by the record and remedied at law.  A coin's redemption is enforced by the contract; a bond's obligation is enforced by the contract once its issuer's date has reached its maturity, and before that by the law, on the record the contract keeps.  One contract has two realisations, with pen and paper and digital, under which the results hold alike~\cite{shapiro2024gc}; here we give what dates add to each.

\mypara{Solvency and clearing}
We prove that a community closed but for one member---e.g., a community bank---clears its debts by redemptions exactly when no other member holds fewer coins and bonds than they owe, of all maturities; and that when it can, members who act on their rights do, in any order, without coordination, once dates advance (Section~\ref{sec:liquidity}).  Clearing with bonds needs the dates to advance, as a bond returns to its issuer only as a coin: the possibility result quantifies over what can happen, and the liveness result over computations in which every member's date advances beyond any bound and members who can redeem or mature eventually do.  What a group can obtain by its own rights now, and whether a run on the bank is met, concern coins, and the results of grassroots currencies~\cite{shapiro2024gc} apply to the coins among the bonds unchanged.

\mypara{Implementation in GLP}
We implement grassroots bonds in GLP~\cite{shapiro2025glp}, a multiagent concurrent logic programming language running on Dart for smartphone deployment, as a program derived from the contract, transaction by transaction.  What dates add to the digital realisation is that each agent keeps its own local date, the issuer judges the maturity of one of their bonds by their own date and that date is recorded at the bearer, and an escrow is released by the date of the escrow agent (Section~\ref{sec:implementation}).  The implementation is demonstrated by a village market of six agents and an escrow agent that exercises symmetric mutual credit, a loan with interest, payments, a portfolio swap, escrow, redemption, sale of debt, and the maturing of a bond.  Making the implementation secure against major faults---the loss of the private key or the machine---using identity and state custodians is the subject of prior work~\cite{eitan2026securing}.

\mypara{Paper structure}
Section~\ref{sec:bonds} presents grassroots bonds with pen and paper: the paper bond and its maturity, holdings and dates, the operations, and the loan and the credit line.  Section~\ref{sec:contract} presents their grassroots social contract as the contract of grassroots currencies with the Date and Escrow clauses: the modality of each, the paper-and-pen realisation, and the transactions with dates the clauses compile to.  Section~\ref{sec:instruments} gives the instruments that are bond swaps, and Section~\ref{sec:escrow} those that employ an escrow agent, the exchange of several lots at once among them.  Section~\ref{sec:liquidity} extends the liquidity ratios to bonds and states the clearing results.  Section~\ref{sec:implementation} gives the digital realisation---what dates add to it, and the village market run on the program---and Section~\ref{sec:conclusion} concludes.  The appendices carry the act schemas and the proof that the contract is grassroots (Appendix~\ref{app:schemas}); the deferred proofs (Appendix~\ref{app:proofs}); the further escrow instruments (Appendix~\ref{app:escrow-more}); the worked example of the ratios (Appendix~\ref{app:liquidity-more}); the community bank with bonds (Appendix~\ref{sec:community-bank}); the derivation of the program from the transactions (Appendix~\ref{app:bond-agent}); and related work (Appendix~\ref{sec:related-work}).

\section{Grassroots Bonds}\label{sec:bonds}

Grassroots currencies were introduced~\cite{shapiro2024gc} and implemented~\cite{lewis2023grassroots} as cryptocurrencies, and are presented in~\cite{shapiro2024gc} as a physical currency, realised with pen and paper: a grassroots coin is an IOU---inside money~\cite{cavalcanti1999inside}---issued freely by its issuer, whose offerings and holdings back it, and every transaction is a voluntary act of all its participants, with the rights and duties of the parties specified by the coins.  Here we add a second instrument, the bond---a debt of its issuer due at a later date---and present it in the same way: first as a physical currency, then, from Section~\ref{sec:contract} on, as a grassroots social contract that extends the contract of grassroots currencies, whose clauses compile to volition-guarded multiagent atomic transactions~\cite{lewis2026volitional,shapiro2026formalising}.  In the following, formal statements can be skipped in favour of their adjacent informal readings.

\mypara{Persons}
A person---natural or legal, with the gender-neutral pronoun ``they''---is identified by their name and signature.  We assume a potentially infinite set of persons $\Pi$; we consider only finite subsets of $\Pi$.

A grassroots coin is a unit of its issuer's debt, due now, backed by the issuer's offerings and by the coins and bonds they hold~\cite{shapiro2024gc}, and a grassroots bond is such a unit due at a later date.  The two are distinct instruments: a coin carries no date, and a bond is minted with a date later than the date of its issuer.

\begin{definition}[Grassroots Coins and Bonds]\label{def:bonds}
A \temph{grassroots $p$-coin}, denoted \textcent$_p$, is a unit of debt issued by $p \in \Pi$, due now.  A \temph{grassroots $p$-bond with maturity date $d$}, denoted \textcent$_{p,d}$, is such a unit due at the date $d \in \calN$.  We write \textcent$^k_p$ and \textcent$^k_{p,d}$ for multisets of $k$ of each.
\end{definition}

In a practical implementation, each $p$-coin and $p$-bond will have a serial number and value and be signed by $p$.  For simplicity of exposition the coins and bonds employed here can form multisets and are implicitly with denomination $1$.

Within the pen-and-paper scenario, we assume that signatures cannot be forged and that signed papers cannot be duplicated (``double-spent''). A pen-and-paper grassroots $p$-bond---a \emph{paper bond}---is the paper coin of grassroots currencies~\cite{shapiro2024gc} with a maturity date written on it: a promissory note, written on paper and signed by $p$, stating the issuer $p$, the maturity date $d$, and three standing obligations:
\begin{enumerate}
    \item $p$ prices their offerings in their own coins and accepts them in payment at these prices;
    \item $p$ prices every coin and bond they hold, of any issuer and maturity, 1-for-1 against their own coins, on standing offer; and
    \item once the date of $p$ has reached the maturity date, the bearer may present the bond to $p$ and take a $p$-coin for it.
\end{enumerate}
Obligations~1 and~2 are those of a paper coin~\cite{shapiro2024gc}, and obligation~3 is the bond's.  A grassroots coin is a legally binding debt of its issuer, due now.  Its bearer may present it to its issuer and take, at the written 1-for-1 price, any coin or bond the issuer holds, of the bearer's choice.  A bearer may take a bond rather than a coin: a debt due later of a person the bearer trusts can be worth more than a debt due now of one the bearer trusts less.  Taking an instrument the bearer issued is set-off: it returns the bearer's own debt, indisputable at par.  A presentation left unfulfilled evidences the issuer's default, which in turn may cause a ``run on the bank'' on $p$ and if $p$ turns out to be insolvent, the depreciation of the market value of $p$'s coins and bonds.  A coin or a bond is transferred by its holder signing it over to the new holder and delivering it, so any two persons transfer any coins and bonds between themselves, with no third party; the signatures on the paper record its provenance.

\mypara{Holdings and dates}
The state of a person $p$ comprises $c_p$, the multiset of coins and bonds in $p$'s possession (initially empty); $d_p^* \in \calN$, $p$'s local current date (initially $0$), advanced by the calendar; and the date warranties $p$ has received (Section~\ref{sec:contract}), which play no part in the results and are omitted from the notation.  A bond is minted with a date later than the date of its issuer, so no bond is born mature.  A bond \textcent$_{q,d}$ is \emph{mature} iff $d \le d_q^*$, and its bearer may then take a $q$-coin for it: whether $q$ owes yet is a fact about $q$, so maturity is settled by $q$'s date and the dates of the holders play no part.  A coin carries no date, so no act of an issuer makes or unmakes one.

A transaction changes the states of its participants and no one else's.  The four operations of grassroots currencies---mint, voluntary swap, pay, and redeem---gain advance-date and mature, and are the acts of the grassroots social contract of Section~\ref{sec:contract}, whose schemas compile to volition-guarded multiagent atomic transactions over the states above (Section~\ref{sec:transactions}):

\begin{enumerate}
    \item \emph{Mint}: a person may create, at will, coins of their own, of any amount, or bonds of their own, of any amount and of any maturity later than their own date---with pen and paper, by writing and signing $k$ papers;
    \item \emph{Advance-date}: local time advances mechanically, by the calendar;
    \item \emph{Voluntary swap}: two persons may exchange coins and bonds by mutual consent---the parties meet and exchange the agreed papers by simultaneous signing-over and delivery, so the swap is physically atomic;
    \item \emph{Pay}: a holder may pay coins of the recipient to the recipient, at will---delivering the recipient's coins to the recipient;
    \item \emph{Redeem}: the bearer of a coin may present it to its issuer, at will, and take for it any coin or bond the issuer holds---the issuer takes the coin back and gives an instrument of the bearer's choice among those they hold;
    \item \emph{Mature}: the bearer of a matured bond may present it to its issuer, at will, and take a coin of the issuer for it.
\end{enumerate}

Only coins are paid.  A coin carries no date, so its issuer can neither refuse a payment of their own coins nor refuse a presentation of one, as obligations~1 and~2 require, and no date is read in either.  Redemption generalises the coin-for-coin redemption of grassroots currencies~\cite{shapiro2024gc} by allowing the bearer of a coin to choose any coin \emph{or bond} held by the issuer of the coin.  A bond is minted, swapped, escrowed and sold, and an issuer may fulfil a presentation with one; it is not paid, and once it is mature its bearer may redeem it against a coin of its issuer.  The issuer's date is the precondition of that act and is warranted there, and in no other (Section~\ref{sec:contract}); it is a machine precondition and not a guard, so the issuer cannot refuse, as obligation~3 requires.  The holder of a bond may sell the debt to a third party, as shown in Section~\ref{sec:instruments}.  A coin or a bond that reaches its issuer is self-held; whether the issuer keeps or cancels the paper is immaterial, as self-held instruments carry no value (Section~\ref{sec:liquidity}).

\mypara{Voluntary transactions and the obligations of the coins and bonds}
Every transaction is a voluntary act of all its participants: minting and advancing one's date are acts of one person; a swap, a payment, a redemption and a maturing take place at a meeting of the two, each performing their part.  A voluntary swap is a contract, made at the meeting; either party may decline it.  A payment, a redemption and a maturing are the bearer's rights under the written obligations: the recipient of a payment must accept their own coins (obligation~1), the issuer must fulfil a presentation of their coin (obligation~2), and the issuer must give one of their coins for a matured bond of their own (obligation~3); declining any of the three is a breach of the terms, actionable at law.  The issuer assumed these obligations voluntarily, by minting the coin or the bond.  The coins and bonds do not determine whether persons act on their rights, nor whether they advance their dates; the analysis of Section~\ref{sec:solvency} makes both explicit assumptions: dates advance, and persons who can redeem or mature eventually do.

\mypara{Financial instruments as swaps of coins and bonds}
Loans, credit lines, sale of debt, collateral, escrow, and further financial instruments are realised as voluntary swaps of grassroots coins and bonds, and as escrow agreements (Sections~\ref{sec:instruments} and~\ref{sec:escrow}).  Two are used repeatedly in this paper.  A \emph{loan} from lender $p$ to borrower $q$ is a swap of freshly-minted $p$-coins against $q$-bonds maturing later; a face value exceeding the grassroots coins lent expresses interest.  A \emph{credit line} from $p$ to $q$ is a standing commitment, held by an escrow agent, to advance $p$-coins against $q$-bonds up to an agreed limit, with drawn amounts carrying interest and undrawn capacity returned at expiry.  A holder of a $q$-bond who loses trust in $q$ may sell the debt in the free market, at a discount, as with coins~\cite{shapiro2024gc}.

\mypara{Local dates}
As there is no global clock, each person's calendar is their own, and persons may disagree on the date.  A bond matures by the calendar of its issuer, so an issuer who does not advance their date does not yet owe on their bonds.  The contract obliges each party to keep their date at the current date, and the issuer warrants their date when a matured bond of their own is taken for one of their coins (Section~\ref{sec:contract}): an issuer whose date is behind the calendar breaches the warranty there, and the record of the act witnesses the breach.  The value of a coin or a bond rests on trust in its issuer, and an issuer whose date lags loses it: their bonds never mature, their debts never fall due, and others price their currency accordingly or decline to hold it at all.  In the ratios of Section~\ref{sec:liquidity} each person classifies their holdings against their own date, so a ratio is the view of the person computing it.  For an option, exercise is judged by the escrow agent's date (Section~\ref{sec:escrow}), so the parties need agree only with the escrow agent.  We assume no bound on the disagreement: it delays a transaction and never compels one.

\section{The Grassroots Social Contract}\label{sec:contract}

A grassroots social contract~\cite{shapiro2026formalising} is a finite list of clauses its parties undertake towards one another.  A clause that describes an act is written as an \emph{act schema}~\cite{shapiro2026formalising}: the act's roles, which of them must will it, and the atoms it requires, forbids, adds and deletes at each role; the volition-guarded transactions~\cite{lewis2026volitional} realising the contract are compiled from the schemas, three syntactic conditions on the schemas suffice for the protocol so realised to be volitionally grassroots, and each clause is \emph{enforced}, \emph{attested} or \emph{undertaken} by the contract according to what its runs determine and record.  Grassroots currencies are specified by a contract of five clauses~\cite{shapiro2024gc}.  Here we write the contract of grassroots bonds as that contract with two further clauses, Date and Escrow: its clauses and the modality of each, its paper-and-pen realisation, and the transactions with dates its schemas compile to.  Its act schemas, and the proof that it is grassroots, are in Appendix~\ref{app:schemas}, and the program every participant runs is given in Section~\ref{sec:implementation}.

\subsection{The Clauses}\label{sec:clauses}

\mypara{Definitions}  A party is a person, natural or legal.  A coin names its issuer and is due now; a bond names a maturity date besides and is due then.  Each party keeps a date.  A bond is mature when its maturity date is at most the date of its issuer.

\mypara{The contract}
\begin{enumerate}
\item \emph{Minting.}  A party may issue coins of their own, of any amount, and bonds of their own, of any amount and of a maturity later than their own date, at their own election.
\item \emph{Date.}  A party keeps their date at the current date, and warrants it whenever a matured bond of their own is taken for one of their coins.
\item \emph{Exchange.}  Two parties may exchange coins and bonds, by an act of both.
\item \emph{Payment.}  A party holding coins of another party may pay them to that party, at their own election, and that party accepts them.
\item \emph{Redemption.}  A party holding a coin of another party may present it to that party, at their own election, and take for it any coin or bond that party holds; and a party holding a matured bond of another party may present it to that party, at their own election, and take a coin of that party for it.
\item \emph{Pricing.}  A party prices their offerings in their own coins and accepts them in payment at those prices.
\item \emph{Escrow.}  A party may deposit coins and bonds with a second party for a third.  The second party releases them to the third or returns them to the first as the agreed condition is met or fails, each transfer by an act of the second party and the party receiving.
\end{enumerate}

\noindent Clauses 1, 3, 4, 5 and 6 are the clauses of grassroots currencies~\cite{shapiro2024gc}, with a bond beside the coin where the act allows one: minting is of coins or of bonds, an exchange carries both, and redemption takes either.  Clauses 2 and 7 are new, and so is the second act of clause~5, by which a matured bond is taken for a coin of its issuer.  Clauses 1--5 and 7 describe acts; clause 6 describes none.  An exchange in which each party gives coins of their own issue opens a mutual credit line; one in which a party gives coins and the other their own bonds of larger face value maturing later is a loan (Section~\ref{sec:bonds}).  The second party of clause~7 is the escrow agent.

\subsection{The Modality of Each Clause}\label{sec:modalities}

A clause of a grassroots social contract is \emph{enforced}, if performing it is an act of the contract and every correct run of it keeps the clause; \emph{attested}, if a party may breach it, with the contract recording the breach; or \emph{undertaken}, if the contract records who gave the undertaking, but performing the undertaking is carried out outside the contract~\cite{shapiro2026formalising}.  The definitions, and \emph{witnessed against}, on which the last two rest, are that paper's.

Clauses 1, 3, 4 and 5 are enforced and clause 6 is undertaken, as in grassroots currencies~\cite{shapiro2024gc}, and the bond does not change the argument: each act permits an act and obliges nothing beyond it, so the runs determine its clause and every correct run keeps it; ``that party accepts them'' of clause~4 is that pay is guarded by the payer and adds the coins at the payee, and the issuer's fulfilment of a presentation in clause~5 is that redeem is guarded by the bearer and deletes the instrument taken at the issuer, and that mature is guarded by the bearer and adds the issuer's coin at the bearer (Appendix~\ref{app:schemas}), so neither the payee nor the issuer is asked; and the swap opening a mutual credit line puts a coin of each party in the other's holding, the record of the pricing undertaking.  Clause~7 is enforced: each of its three transfers permits an act, guarded by the two parties to the transfer, and obliges nothing beyond it.

Clause~2 is attested.  Failing to advance one's date is not an act, so it breaches nothing by itself; the warranty is given in the course of mature, and an issuer whose date is behind the calendar breaches it there.  Mature adds $\mathit{dated}(q,t)$ in the state of the bearer $p$, naming the issuer $q$ and the date they warranted, and no other schema of the contract reads or warrants a date (Appendix~\ref{app:schemas}), so the one transaction that may breach the clause is witnessed against the issuer, and the warranty is legally enforceable on that record.

Whether a bond is mature is subject to its issuer advancing their date, so the contract enforces the maturing of a bond only against the date its issuer warrants.  Once the date of $q$ has reached the maturity of a $q$-bond, its bearer takes a $q$-coin for it at will, and that act is enforced, as clause~5.  Until then no transaction of the contract gives a coin for it, and an issuer whose date lags owes nothing yet on their bonds, as none of them has matured; if the calendar is past the maturity, $q$ is in breach of clause~2, witnessed by the record of the next maturing of a $q$-bond, and the warranty is what makes that a breach rather than a fact about the issuer.  No record witnesses an issuer who holds their date back far enough that no bond of their own matures; their debts never fall due, and others price their currency accordingly or decline to hold it (Section~\ref{sec:bonds}).  The obligation of a bond as a whole---to give a coin for it once its maturity date has come---is therefore attested; the obligation of a coin is enforced.

\subsection{The Paper-and-Pen Realisation}\label{sec:paper-realisation-sub}

The contract has two realisations.  In the first, people read the clauses and follow them with paper and pen, and no machine takes part.

\begin{definition}[Paper-and-Pen Realisation]\label{def:paper-realisation}
The \temph{paper-and-pen realisation} of the contract is its performance by its parties on signed notes: a party mints by writing and signing notes of their own, a coin note or a bond note stating its maturity date (clause~1); a party keeps a calendar, and states their date when they give one of their coin notes for a matured bond note of their own (clause~2); two parties exchange notes by signing them over and delivering them to each other at a meeting, both at once (clause~3); a party pays by signing over and delivering to another party coin notes that party issued (clause~4); a party presents a coin note to its issuer, who takes it back and gives in return a note of the presenter's choice among those they hold, and presents a matured bond note to its issuer, who takes it back and gives a coin note of their own (clause~5); a party prices their offerings in their own notes (clause~6); and a party deposits notes with the escrow agent, who releases or returns them, each transfer by signing over and delivering at a meeting of the two (clause~7).  A party's holding is the notes in their possession, their date is that of their calendar, the warranties they have received are the dates stated to them, and their consents are their own.
\end{definition}

Under this realisation, safety and liveness are matters of law and of persons.  An act the clauses do not permit is a breach, actionable at law, and so is declining a payment or a presentation, or letting one's date lag (Section~\ref{sec:modalities}); and whoever persistently wills a permitted act with a willing counterparty eventually performs it, and the calendar advances---assumptions about persons.  In the digital realisation (Section~\ref{sec:implementation}) both are properties of the program every participant runs, proven for Grassroots Flash~\cite{lewis2023grassroots}.  The results of this paper are theorems about the transactions the contract compiles to (Section~\ref{sec:transactions}); the acts of the paper-and-pen realisation are those transactions performed on notes, so the results hold under both realisations.

\subsection{The Transactions}\label{sec:transactions}

The digital realisation begins by compiling the act schemas to volition-guarded multiagent atomic transactions~\cite{lewis2026volitional,shapiro2026formalising}.  Each agent is a person operating a machine, and here the machine state of a party $p$ is their holdings $c_p$, a multiset of coins and bonds, initially empty; their date $d_p^*$, initially $0$; and the warranties $w_p$ they have received, a set of pairs of an issuer and a date, initially empty.  A transaction changes the machine states of its participants and no one else's, is enabled when its precondition holds, and is \emph{guarded} by the persons among its participants whose consent it takes: it is taken only when every guard is willing.  Transactions that are the same act in different configurations are \emph{equivalent}.  A run is \emph{safe} if every step is a transaction of the protocol, \emph{live} if a class of equivalent transactions that is enabled and willed by its guards is eventually taken, and \emph{correct} if it is both; a protocol is \emph{volitionally grassroots} if any group of persons may operate it on its own, and two groups become connected only by a transaction willed by a member of each~\cite{lewis2026volitional}.  Multiset sum and difference are written $\uplus$ and $\setminus$.  The transactions of grassroots currencies~\cite{shapiro2024gc} are the four acts over holdings of coins; the acts of grassroots bonds are the following transactions, in which mint takes two forms, a swap carries coins and bonds alike, redeem takes a coin or a bond, and mature, advance-date and the three escrow transfers are new.

\begin{definition}[The Transactions of Grassroots Bonds]\label{def:gb-transactions}
For persons $p \ne q \in \Pi$, $e \in \Pi$ distinct from $p$ and $q$, $u, v, r \in \Pi$, maturity dates $d, d' \in \calN$, a date $t \in \calN$, and amounts $k, k' \ge 1$, the \temph{transactions of grassroots bonds} are the volition-guarded transactions $(c \rightarrow c', Q')$ below, over the participants named, in which $c$ satisfies the precondition stated and every component of a participant's state not named is unchanged.  We write \textcent$^k_{u,\ast}$ for $k$ instruments of $u$ of one kind and maturity, coins or bonds, where the transaction does not distinguish them:
\begin{enumerate}
    \item \temph{Mint$_p(k)$} and \temph{Mint$_p(d,k)$}, over $\{p\}$, guarded by $\{p\}$: the second with precondition $d > d^*_p$, so no bond is born mature; $c'_p = c_p \uplus \text{\textcent}^k_p$ and $c'_p = c_p \uplus \text{\textcent}^k_{p,d}$ respectively.
    \item \temph{Advance$_p(t)$}, over $\{p\}$, guarded by $\emptyset$: precondition $t > d^*_p$; $d'^*_p = t$.
    \item \temph{Swap$_{p,q}(\text{\textcent}^k_{u,\ast}, \text{\textcent}^{k'}_{v,\ast})$}, over $\{p,q\}$, guarded by $\{p,q\}$: precondition $\text{\textcent}^k_{u,\ast} \subseteq c_p$ and $\text{\textcent}^{k'}_{v,\ast} \subseteq c_q$; $c'_p = c_p \setminus \text{\textcent}^k_{u,\ast} \uplus \text{\textcent}^{k'}_{v,\ast}$ and $c'_q = c_q \setminus \text{\textcent}^{k'}_{v,\ast} \uplus \text{\textcent}^k_{u,\ast}$.
    \item \temph{Pay$_{p,q}(k)$}, over $\{p,q\}$, guarded by $\{p\}$: precondition $\text{\textcent}^k_q \subseteq c_p$; $c'_p = c_p \setminus \text{\textcent}^k_q$ and $c'_q = c_q \uplus \text{\textcent}^k_q$.  Only coins are paid.
    \item \temph{Redeem$_{p,q}(r,\ast)$}, over $\{p,q\}$, guarded by $\{p\}$: precondition $\text{\textcent}_q \in c_p$ and $\text{\textcent}_{r,\ast} \in c_q$; $c'_p = c_p \setminus \text{\textcent}_q \uplus \text{\textcent}_{r,\ast}$ and $c'_q = c_q \setminus \text{\textcent}_{r,\ast} \uplus \text{\textcent}_q$.  The bearer takes a coin or a bond, at their own choice.
    \item \temph{Mature$_{p,q}(d)$}, over $\{p,q\}$, guarded by $\{p\}$: precondition $\text{\textcent}_{q,d} \in c_p$ and $d \le d^*_q$; $c'_p = c_p \setminus \text{\textcent}_{q,d} \uplus \text{\textcent}_q$ and $w'_p = w_p \cup \{(q,d^*_q)\}$.
    \item \temph{Deposit$_{p,e}(\text{\textcent}^k_{u,\ast})$}, over $\{p,e\}$, guarded by $\{p,e\}$: precondition $\text{\textcent}^k_{u,\ast} \subseteq c_p$; $c'_p = c_p \setminus \text{\textcent}^k_{u,\ast}$ and $c'_e = c_e \uplus \text{\textcent}^k_{u,\ast}$.
    \item \temph{Release$_{e,q}(\text{\textcent}^k_{u,\ast})$} and \temph{Return$_{e,p}(\text{\textcent}^k_{u,\ast})$}, over $\{e,q\}$ and $\{e,p\}$, guarded by both participants: precondition $\text{\textcent}^k_{u,\ast} \subseteq c_e$; the instruments are deleted at $e$ and added at the other participant.
\end{enumerate}
Two transactions are \temph{equivalent} iff they have the same name and arguments, differing only in the configurations $c$ and $c'$.
\end{definition}

Each is a schema of Appendix~\ref{app:schemas} at a binding: the participants are the roles, the guards are the guarding roles, the precondition is what the roles require, and the effect is what they add and delete; the equivalence is that of one schema at one binding.  The \emph{grassroots bonds} protocol $\mathit{GB}$ is the protocol over these transactions~\cite{lewis2026volitional}: the protocol realising the contract~\cite{shapiro2026formalising}.  Mint, pay, redeem and mature are guarded by the initiator---an obligation of the counterparty under the coins and bonds, so their volition is not required; a swap and each escrow transfer are contracts, guarded by both; advance-date is unguarded, as local time advances mechanically.  In a correct run, a class of transactions that is enabled and willed by its guards is eventually taken, so every date advances beyond any bound and every enabled redemption and maturing that its bearer wills is eventually taken; this is the formal counterpart of the assumptions of Section~\ref{sec:liquidity}.

\begin{proposition}[Conservation of Money and Chain Redemption]\label{prop:conservation}
Conservation of Money and Chain Redemption of grassroots currencies~\cite{shapiro2024gc} hold for the transactions of grassroots bonds: in any safe run, the debts of $p$ in any configuration are exactly those $p$ has minted so far, each of them a $p$-coin or a $p$-bond; and if each person of a sequence of distinct persons holds a coin of the next, the first obtains a coin of the last by redeeming along the sequence.
\end{proposition}

\begin{proof}
Mint adds new debts of $p$ at $p$; swap, pay, redeem, deposit, release and return move coins and bonds between two agents without creating or destroying them; mature replaces a matured $p$-bond by a $p$-coin, the same debtor and the same face value; advance-date changes no holding, and neither does a warranty or a change of volition.  Hence the debts of $p$ across all agents are those $p$ minted.  For the sequence, $p_0$ redeems the $p_1$-coin they hold against the $p_2$-coin held by $p_1$, so the sequence without $p_1$ satisfies the hypothesis with one person fewer, and induction on its length completes the proof.
\end{proof}

\mypara{Grassroots}
Bonds among the members of a group require nothing outside the group, so any group may operate grassroots bonds on its own, and two groups may transact once their members meet: the grassroots bonds protocol is grassroots, and volitionally grassroots.

\begin{corollary}[Volitionally Grassroots]\label{cor:bonds-grassroots}
Grassroots bonds are volitionally grassroots.
\end{corollary}

\begin{proof}
The contract is syntactically grassroots (Appendix~\ref{app:schemas}); the protocol realising a syntactically grassroots contract is volitionally grassroots~\cite{shapiro2026formalising}.
\end{proof}

\section{Financial Instruments via the Voluntary Swap of Bonds}\label{sec:instruments}

Financial instruments are expressed via the voluntary swap of grassroots bonds and via escrow agreements.  The instruments below are single swaps; each is specified as a transaction of the model, and further instruments are combinations of them.  A swap exchanges one lot---coins of one issuer, or bonds of one issuer and maturity---against another (Section~\ref{sec:contract}); the instruments that exchange several lots at once, and those needing a third party---to hold bonds, or to judge an outside event---are developed in Section~\ref{sec:escrow}, where the escrow agent appears as an agent of the model.  Each instrument below is a swap between agents $p$ and $q$, guarded by the two agents, specified by the instruments $x$ transferred from $p$ to $q$ and the instruments $y$ transferred from $q$ to $p$.  Symmetric mutual credit, asymmetric mutual credit, and the sale of debt are swaps of coins~\cite{shapiro2024gc}; what a maturity date adds is the following.

An \emph{oracle} is an agent the parties agree to trust for the value of a quantity outside the protocol---a reference interest rate, a delivery, the default of a third person.  An oracle holds no coins or bonds and takes no transaction: it attests, and the agents who trust it act on the attestation.  The oracle is the information source of the escrow constructions of Section~\ref{sec:escrow}, and the trust the parties place in it is part of the agreement.

\begin{enumerate}

\item \textbf{Zero-coupon loan.}  Lender $p$ mints $k' < k$ coins, borrower $q$ mints a bond maturing at $d$:  $x = \text{\textcent}^{k'}_p$, $y = \text{\textcent}^k_{q,d}$.  The premium of asymmetric mutual credit~\cite{shapiro2024gc} is interest here: a bond is redeemed only once it is mature, so the loan cannot be recalled before its date, whereas the borrower may repay it early, by redeeming the lender's coins against their own bond.

\item \textbf{Sale of debt before maturity.}  $p$ sells $r$-bonds maturing at $d$ to $q$ for $k' < k$ $q$-coins:  $x = \text{\textcent}^k_{r,d}$, $y = \text{\textcent}^{k'}_q$.  The debtor $r$ is not a party; the discount prices the time to maturity as well as the doubt about $r$'s solvency.

\item \textbf{Forward contract.}  $p$ and $q$ exchange bonds both maturing at $d$:  $x = \text{\textcent}^k_{p,d}$, $y = \text{\textcent}^{k'}_{q,d}$.  The bonds are swapped at inception and mature together, so at $d$ each party holds a claim on the other of the agreed face value; the amounts due are those of a forward settled at $d$.  The risk differs: between inception and $d$ each party holds the other's debt, so a default in that interval falls on the counterparty, whereas a forward settled at $d$ leaves nothing outstanding until then.
\end{enumerate}

The instruments cover loans (item 1), debt trading (item 2), and derivatives (item 3).  Currency swaps and repurchase agreements (repos) are realised as sequences of two swaps; factoring (selling receivables at a discount) is an instance of sale of debt; a mortgage is a loan combined with collateral.  Loans with payment schedules---balloon and fixed-payment loans---and the interest rate swap exchange several lots at once, and are developed in Section~\ref{sec:escrow} as packaged exchanges, together with credit lines; collateral, guarantees, options, insurance, credit default swaps, and letters of credit are in Appendix~\ref{app:escrow-more}.

\section{Escrow-Based Financial Instruments}\label{sec:escrow}

An escrow agent $e$ holds coins and bonds on behalf of the parties to an escrow agreement and transfers them according to its agreed-upon conditions.  Each escrow transaction is a transfer guarded by the two agents it involves---the agent giving the bonds and the agent receiving them; as $e$ is party to every transfer, each escrow transaction is guarded by $e$ and one of the parties.  Thus $e$ is a volitional agent, whose volitions are determined by the escrow program implementing the agreement: for a date-based condition, by the local date $d_e^*$ of $e$; for an event-based condition, by the judgment of $e$, informed where needed by an oracle agreed by the parties (Section~\ref{sec:instruments}).  The escrow constructions assume that both parties trust $e$ to follow the escrow program and, for event-based conditions, trust $e$ as adjudicator of the condition. We write $\{(a, x),\; (b, y)\}$ for the exchange in which $a$ gives the instruments $x$ and $b$ gives the instruments $y$, guarded by $a$ and $b$; $y = \emptyset$ makes it a one-sided transfer.  The contract has four: the deposit, the release and the return of clause~7, and the payment of clause~4, which carries only coins of the party receiving them.  Coins and bonds therefore do not pass one way from one ordinary party to another; an exchange has two sides.  An exchange or transfer of one lot---coins of one issuer, or bonds of one issuer and maturity---is one act of the contract (Section~\ref{sec:contract}); one of several lots is the corresponding sequence of acts, executed all-or-nothing by the packaged exchange below.  An escrow whose agent is a committee, acting when a threshold of its members agree, is not expressible in the present model, in which every agent is one person; committee agents are left to further work.

\mypara{Mechanism}  An escrow arrangement has a depositor, a beneficiary, and an escrow agent $e$.  The depositor transfers coins or bonds to $e$, guarded by the depositor and $e$.  Subsequently $e$ either releases them to the beneficiary or returns them to the depositor, each guarded by $e$ and the recipient, as the escrow program determines from the agreed conditions.  In the transfers below, the two named agents are the participants in, and guards of, the transaction.

\begin{enumerate}
\item \textbf{Packaged exchange:}  Parties $p$ and $q$ agree the lots $X$ that $p$ gives, the lots $Y$ that $q$ gives, and an establishment window $T_0$:

\begin{enumerate}
\item \textbf{Deposit.}  $\{(p, X),\; (e, \emptyset)\}$ and $\{(q, Y),\; (e, \emptyset)\}$: each party deposits its lots with $e$, one transfer per lot.
\item \textbf{Release.}  Once every lot of both is deposited, $\{(e, Y),\; (p, \emptyset)\}$ and $\{(e, X),\; (q, \emptyset)\}$: $e$ releases the deposits crosswise.
\item \textbf{Return.}  If the window passes with the lots of a party incomplete, $e$ returns every deposited lot to its depositor.
\end{enumerate}

All lots change hands or none do, and neither party holds the other's lots before both deposits are complete.  The packaged exchange realises the instruments whose sides are several lots:

\begin{enumerate}
\item \textbf{Balloon loan.}  Lender $p$ gives $X = \text{\textcent}^k_p$; borrower $q$ gives $Y = \bigcup_{j=1}^{n} \text{\textcent}^{k_j}_{q,d_j} \cup \text{\textcent}^k_{q,d}$---coupon bonds $\text{\textcent}^{k_j}_{q,d_j}$ and a principal bond $\text{\textcent}^k_{q,d}$.
\item \textbf{Fixed-payment loan.}  Lender $p$ gives $X = \text{\textcent}^k_p$; borrower $q$ gives $Y = \bigcup_{j=1}^{n} \text{\textcent}^{k_j}_{q,d_j}$, bonds with fixed payments $k_j$ at dates $d_j$.
\item \textbf{Interest rate swap.}  $X = \bigcup_j \text{\textcent}^{k_j}_{p,d_j}$, $Y = \bigcup_j \text{\textcent}^{k'_j}_{q,d_j}$, where $p$ pays fixed amounts $k_j$ while $q$ pays amounts $k'_j$ that vary according to a reference rate, attested by the interest-rate oracle.  Periodic settlement exchanges realise the net difference at each date.
\end{enumerate}

\item \textbf{Further escrow-based instruments:}  Collateral, guarantees, options, insurance, credit default swaps, and letters of credit are escrow constructions of the same three transfers---deposit, release, return; they are given in Appendix~\ref{app:escrow-more}.

\item \textbf{Credit line:}  A lender $p$ commits to advance up to $k$ coins to borrower $q$, with periodic interest on drawn credit and line expiry at date $T$.  The agreement specifies the limit $k$, interest rate $\rho$ (assumed to yield integer interest amounts), payment dates $d_1, \ldots, d_n$, and expiry $T$.

\begin{enumerate}
\item \textbf{Establishment.}  $p$ mints $k$ coins and deposits them with escrow agent $e$: $\{(p, \text{\textcent}^k_p),\; (e, \emptyset)\}$.  $e$ tracks the drawn amount $k_d$, initially $0$.
\item \textbf{Draw.}  When $q$ draws $k' \le k - k_d$, $q$ mints a principal bond $\text{\textcent}^{k'}_{q,T}$ (maturing at line expiry) and coupon bonds $\text{\textcent}^{\rho k'}_{q,d_j}$ for each $d_j > d_e^*$, and deposits them with $e$, one transfer per lot: $\{(q, \text{\textcent}^{k'}_{q,T} \cup \bigcup_{d_j > d_e^*} \text{\textcent}^{\rho k'}_{q,d_j}),\; (e, \emptyset)\}$.  Once every lot is deposited, $e$ releases them to $p$, one transfer per lot, $\{(e, \text{\textcent}^{k'}_{q,T} \cup \bigcup_{d_j > d_e^*} \text{\textcent}^{\rho k'}_{q,d_j}),\; (p, \emptyset)\}$, and releases the coins to $q$, $\{(e, \text{\textcent}^{k'}_p),\; (q, \emptyset)\}$---the release of a packaged exchange (item~1) whose lender's side $e$ already holds---and updates the drawn amount $k_d$ to $k_d + k'$.
\item \textbf{Repayment.}  $q$ acquires $p$-coins (e.g.\ by trade with third parties or from their own holdings) and redeems them from $p$ against their principal bond, coin by coin (Section~\ref{sec:contract}, Redemption): repayment is $q$'s right under the written obligations, and $p$ cannot refuse it.  $p$ then deposits the $p$-coins with $e$ via $\{(p, \text{\textcent}^{k_{\text{rep}}}_p),\; (e, \emptyset)\}$, restoring capacity: $e$ updates $k_d$ to $k_d - k_{\text{rep}}$.  Repayment may be partial: any $k_{\text{rep}} \le k_d$ restores capacity by $k_{\text{rep}}$ and leaves $k_d - k_{\text{rep}}$ drawn.  The coupon bonds are $q$'s debt from the draw and remain due at their dates; returning the unaccrued ones is a further transfer from $p$ to $q$.
\item \textbf{Expiry.}  Once $d_e^* \ge T$, $e$ returns the undrawn balance to $p$: $\{(e, \text{\textcent}^{k-k_d}_p),\; (p, \emptyset)\}$.
\end{enumerate}

\end{enumerate}

The credit line is the one construction in which the escrow agent carries state across transactions: the drawn amount $k_d$.  It is not extra state: it is the limit less the $p$-coins $e$ holds.

\begin{proposition}[Credit Line]\label{prop:credit-line}
Let $k$ be the limit, $T$ the expiry, and $k_d$ the drawn amount of the credit line above.  In any safe run in which $e$ follows the escrow program and $p$ and $q$ transact only as the agreement specifies, between the operations of the agreement---each packaged exchange or transfer complete---$e$ holds \textcent$^{k-k_d}_p$, $p$ holds \textcent$^{k_d}_{q,T}$ together with the coupon bonds of the draws, and $0 \le k_d \le k$.  At expiry with $k_d = 0$, $p$ holds \textcent$^k_p$ and no $q$-bond maturing at $T$.
\end{proposition}

The coins $e$ holds and the principal $p$ holds sum to the limit, so $q$ cannot draw beyond $k$ and the principal $p$ holds equals the amount drawn.  The proof is in Appendix~\ref{app:proofs}.

\section{Liquidity Ratios, Solvency, and Clearing}\label{sec:liquidity}\label{sec:solvency}

This section extends the liquidity ratio of grassroots currencies~\cite{shapiro2024gc} to bonds and proves the clearing results for bonds.  The ratios are per person: the cash, quick, and current liquidity ratios of corporate finance apply to grassroots bonds, with bond maturity playing the role of asset class.  The results are per community, stated against its community bank or any designated member: whether its debts, of all maturities, can be extinguished by redemptions, and whether they are, once dates advance and persons act on their rights.  The persons act independently, each by their own rights, so we state and prove the answers in graph terms: the holdings form a graph, and solvency depends on the graph degrees.  What a group can obtain of a member's coins now, and whether a run on the bank is met, are questions about coins, answered in~\cite{shapiro2024gc}; the section ends with how the answers apply to bonds.

\mypara{Liquidity ratios}
For grassroots coins the cash, quick, and current liquidity ratios of corporate finance coincide in one liquidity ratio, as every coin is redeemable now~\cite{shapiro2024gc}.  Here we recover the three ratios exactly, with maturity (a coin, a near-term bond, a long-term bond) playing the role of asset class. Recall that $d_p^*$ denotes the local current date of person $p$ (Section~\ref{sec:bonds}).  Let $\nu_p(q)$ be the number of $p$-coins held by $q$ at present, and $\nu_{p,t}(q)$ the number of $p$-coins and $p$-bonds maturing at or before $t$ held by $q$ at present, both zero if $p = q$, as instruments $p$ holds of their own carry no value.  Let $\delta$ denote a near-term horizon (typically 90 days) and $\Delta$ the operating cycle (typically one year), with $0 < \delta < \Delta$.  The current liabilities of $p$ are their coins and their bonds maturing within the operating cycle, $\sum_{q} \nu_{p,\,d_p^*+\Delta}(q)$---the common denominator of the three ratios of $p$:
{\small
$$
\textit{Cash Ratio: } \frac{\sum_{r} \nu_{r}(p)}{\sum_{q} \nu_{p,\,d_p^*+\Delta}(q)}
\quad
\textit{Quick Ratio: } \frac{\sum_{r} \nu_{r,\,d_p^*+\delta}(p)}{\sum_{q} \nu_{p,\,d_p^*+\Delta}(q)}
\quad
\textit{Current Ratio: } \frac{\sum_{r} \nu_{r,\,d_p^*+\Delta}(p)}{\sum_{q} \nu_{p,\,d_p^*+\Delta}(q)}
$$
}
With the common denominator and progressively wider numerators, the ratios nest: cash ratio $\le$ quick ratio $\le$ current ratio.  A worked example on the village market is in Appendix~\ref{app:liquidity-more}; the creditworthiness measures of grassroots currencies---foreign debt, current account, and velocity---apply unchanged, counted over coins and bonds of all maturities~\cite{shapiro2024gc}.

\mypara{Insolvency and illiquidity}
The ratios separate two failures.  A person is \emph{insolvent} when their liabilities exceed their assets: they cannot repay their debt, and no loan remedies that---the current ratio, all assets against all liabilities of the operating cycle, bears on it.  A solvent person is \emph{illiquid} when the coins they hold cannot meet the claims arriving now, though their assets suffice at their maturities---the cash ratio measures exactly this, its numerator the person's immediate means of meeting redemptions.  For coins the two failures are one, as every asset is a coin; the bond separates them.  Default is sanctioned by the community~\cite{besley1995social}: holders redeem what the defaulter still holds, no one extends them new credit, they lose set-off and the ability to trade in their own currency, and the market value of their coins and bonds depreciates.

\mypara{Computations}
A \emph{configuration} records, for every person, their holdings and their local current date; a \emph{transition} leads from one configuration to the next---a transaction of the contract (Section~\ref{sec:transactions}) taking place; the date warranties play no part in the results and are omitted.  A \emph{computation} is a sequence of configurations, each reached from the previous by a transition; it is \emph{safe} if every step is one, and a transaction is \emph{enabled} in a configuration if its preconditions hold there.  The possibility result below (Proposition~\ref{prop:graph-solvency}) quantifies over safe computations---what can happen; the liveness result (Proposition~\ref{prop:clearing-liveness}) quantifies over computations in which dates advance and persons act on their rights---persons who can redeem or mature eventually do (Definition~\ref{def:eventual-redemption})---and states what must then happen.  In the volitional framework, eventual redemption is the liveness of correct runs (Section~\ref{sec:transactions}): an enabled class that is willed by its guards is eventually taken, and advance-date is willed by no one.  The proofs of the propositions below are in Appendix~\ref{app:proofs}.

\mypara{Graph solvency}
Can a community extinguish all its internal debts?  A debt is extinguished when the coin or the bond returns to its issuer, and a redemption does exactly that: the bearer returns a coin to its issuer and takes in exchange any coin or bond the issuer holds---as of right (Section~\ref{sec:bonds}), without transferring value, and without changing any person's assets net of liabilities.  The propositions quantify over safe computations and do not assume a holder can inspect the debtor's holdings; in practice a debtor incapable of set-off discloses their holdings to the redeemer.  The results are stated relative to a designated person $s \in \Pi$, the \temph{sink}: the member with whom the community's net position is left---e.g., a community bank in a village (Appendix~\ref{sec:community-bank}).  Here we define the debts of a community as a graph, one edge per coin and per bond from holder to issuer, and prove: the community can extinguish all its debts, by redemptions, exactly when every member holds at least as many coins and bonds as they owe (Proposition~\ref{prop:graph-solvency}); and when it can, members who act on their rights do, in any order, without coordination, once dates advance (Proposition~\ref{prop:clearing-liveness}).  Clearing uses Redeem, Mature and Advance-date (Section~\ref{sec:contract}): a Pay transfers value and a swap requires the counterparty's consent, while Advance-date brings a debtor's date up to the maturity of their bonds, and Mature turns a matured bond into a coin of its issuer at the bearer's will.  For coins, clearing uses Redeem~\cite{shapiro2024gc}; a bond returns to its issuer only as a coin, so its bearer takes a coin for it first, which needs the debtor's date to reach its maturity.

\begin{definition}[Holding Graph, Sink-Closed]\label{def:holding-graph}
Given a configuration and a designated person $s \in \Pi$, the \temph{sink}, the \temph{holding graph} is the directed multigraph with the persons as vertices and an edge $u \rightarrow v$ for each $v$-coin and each $v$-bond that person $u \ne v$ holds.  A set of vertices containing $s$ is \temph{sink-closed} in a directed multigraph over the persons if every edge between two vertices other than $s$ has both endpoints in the set or neither.
\end{definition}

We write $a_p$ and $\ell_p$ for the out-degree and in-degree of person $p$ in the holding graph: their assets and liabilities in circulation, counted in coins and bonds of all maturities.  A sink-closed set of persons is a community whose only claims across its boundary are with the sink; the sink is a member of every such community.  A Redeem transaction in which the debtor is not $s$ and the instrument taken is not issued by the debtor is a \emph{clearing redemption}.

\begin{definition}[Clearing]\label{def:clearing}
Given a configuration and a sink-closed set $V$ of the holding graph, the holdings of $V$ \temph{clear} if some finite safe computation of clearing redemptions among members of $V$, Mature transactions of members of $V$, and Advance-date transitions of members of $V$ (Section~\ref{sec:contract}) leads to a configuration in which every coin and bond of an issuer in $V$ other than $s$ is held by its issuer.
\end{definition}

\begin{proposition}[Graph Solvency]\label{prop:graph-solvency}
The holdings of a sink-closed set $V$ clear if and only if $a_p \ge \ell_p$ for every member $p$ of $V$ other than $s$, and cleared holdings leave each such member holding exactly $a_p - \ell_p$ coins and bonds of the sink.
\end{proposition}

The liveness counterpart replaces ``some computation clears'' by ``every computation in which dates advance and persons act on their rights clears''.  In the pen-and-paper scenario, that the calendar advances and that persons act are assumptions about persons.  The digital realisation proves them of the protocol (Section~\ref{sec:implementation}).

\begin{definition}[Eventual Redemption]\label{def:eventual-redemption}
An infinite safe computation of clearing redemptions among members of a sink-closed set $V$, Mature transactions of members of $V$, and Advance-date transitions of members of $V$ satisfies \temph{eventual redemption} if (\ia)~the local current date of every member of $V$ advances beyond any bound, and (\ib)~no clearing redemption among members of $V$ and no Mature transaction of a member of $V$ is enabled in every configuration of some suffix in which neither occurs.
\end{definition}

\begin{proposition}[Clearing Liveness]\label{prop:clearing-liveness}
If $a_p \ge \ell_p$ for every member $p$ of $V$ other than $s$, then every infinite safe computation of clearing redemptions among members of $V$, Mature transactions of members of $V$, and Advance-date transitions of members of $V$ that satisfies eventual redemption reaches cleared holdings of $V$, after at most $\sum_{p \in V \setminus \{s\}} \ell_p$ redemptions.
\end{proposition}

Proposition~\ref{prop:graph-solvency} has three special cases: (\ia)~$V$ is all persons: the economy clears exactly when every person's assets in coins and bonds cover their liabilities in them; when every bond matures within the operating cycle, the condition reads: every person has a current ratio of at least~1.  (\ib)~$V$ has no claims on or by the sink: every coin and bond then counts once as an asset and once as a liability within $V$, so the condition holds only with equality throughout---the community clears exactly when every member holds precisely as much as they owe, and each member ends holding only their own coins and bonds.  (\ic)~$V$ is a cycle of mutual debt (with the sink an isolated member): the cycle clears with no coin or bond of the sink, as chain redemption around the cycle ends in a set-off.  A sink-closed set is the sink together with any union of connected components of the holding graph without it, so one sink serves many communities, each clearing on its own.  The propositions concern one configuration; over time, liabilities grow with interest, which enters through minting and swaps, not through clearing.

\mypara{Coins among the bonds}
Two further questions of grassroots currencies~\cite{shapiro2024gc} concern coins, and their answers apply to the coins among the bonds.  The \emph{coin holding graph} is the subgraph of the holding graph consisting of the edges of coins.  What a group $A$ of persons can obtain of the sink's coins by its own redemptions, now---with no date advancing, as advancing the dates of others is outside the group's control---is the maximum flow from $A$ to $s$ in the coin holding graph: a redemption that takes a bond obtains no coin and consumes the coin surrendered, so the group does best by redemptions of coins against coins, which are the redemptions of grassroots currencies over the coin holding graph, and the maximum-flow proposition of~\cite{shapiro2024gc} applies.  The flow counts the coins the group holds, and a member holding a matured bond takes a coin of its issuer for it first, which adds that edge to the coin holding graph.  For a single person, the numerator of their cash ratio counts the coins they hold---their out-degree in the coin holding graph---and the maximum flow is the number of $s$-coins they can reach.  A run on the bank---all coins of a person presented within a single day, in which no bond is taken for a coin---is met in full exactly when the coins they hold cover the coins they owe, their bonds being neither an asset nor a liability in it: the run is a computation of grassroots currencies over the coin holding graph, and the run-on-the-bank proposition of~\cite{shapiro2024gc} applies.  When it is not met, the presentations that cannot be met stand unfulfilled, evidencing the person's default (Section~\ref{sec:bonds}).

\section{Digital Realisation}\label{sec:implementation}

The digital realisation of grassroots currencies---the issuer as the sole point of trust, the digital social contract every participant runs, finality by the issuer's signature and recovery by custodians, and the program derived from the contract, transaction by transaction---is given in~\cite{shapiro2024gc}.  Here we give what dates and escrow add to it, and run the bonds village market on the program.

\mypara{Volitions and guards}
In the pen-and-paper scenario, a person consents by taking part in the transaction; digitally, standing consents are recorded.  The assumptions of Section~\ref{sec:liquidity}---that dates advance and that persons act on their rights---become the correctness requirement of the framework: in a correct run, a transaction that is enabled and willed by its guards is eventually taken, and advance-date is willed by no one~\cite{lewis2026volitional}.  The protocol so specified is \emph{volitionally grassroots}: two groups become coupled only by a first transaction willed by a member of each (Corollary~\ref{cor:bonds-grassroots}).

\mypara{The digital social contract}
The contract of Section~\ref{sec:contract} is realised by the \emph{digital social contract} implementing the currency: the multiagent program that every participant runs, its instances communicating only with each other~\cite{cardelli2020digital}.  Its enforced clauses are self-executing, as for coins~\cite{shapiro2024gc}, where the issuer signs every transaction in their own coins and bonds, and here applies their own date to them: a bond presented for a coin of its issuer is mature by the issuer's date or it is not, and no party needs to see another's date.  The date warranty (clause~2) is attested: each maturing records the issuer's date at the bearer.  The escrow transfers (clause~7) are guarded by the escrow agent and the party receiving, and the escrow program of the agent wills release or return as its own date, or the attestation of an oracle, determines (Section~\ref{sec:escrow}).  The pricing undertaking (clause~6) is performed outside the contract.  Finality by the issuer's signature~\cite{lewis2023grassroots}, recovery from the loss of a key or a machine by identity and state custodians~\cite{eitan2026securing}, enforcement by mutual attestation, default and identity escape, the digital platforms a currency can run on, and privacy are as for coins~\cite{shapiro2024gc}.

\mypara{Implementation in GLP}
Grassroots Logic Programs (GLP)~\cite{shapiro2025glp} is a multiagent, concurrent, logic programming language running on Dart for smartphone deployment.  A grassroots specification is realised as a GLP program, which people load and activate voluntarily on their devices.  Here we implement the contract of Section~\ref{sec:contract} as a program in vGLP~\cite{shapiro2026volition}, GLP with volition guards, with one volition-guarded clause per transaction of Section~\ref{sec:transactions}, and compile it by the vGLP compiler to the GLP program every participant runs, together with the interface constructs the compilation derives from the clauses; the derivation, transaction by transaction, is in Appendix~\ref{app:bond-agent}.  What dates add to the digital realisation is three things: each agent keeps its own local date, advanced by Advance-date from its calendar, with no global clock; the issuer tests a bond presented for one of their coins by their own date, an act they must fulfil, and the date they warrant is recorded at the bearer; and an escrow agent releases or returns what is deposited with it as the escrow program wills, for a date-based condition by the escrow agent's own date.  The Grassroots Super-App~\cite{shapiro2026gsg} runs grassroots currencies as a mini-app and renders its interface~\cite{shapiro2024gc}; running the program of grassroots bonds as a mini-app is further work, and the run below is of the compiled program with scripted persons in place of the interface.

\subsection{Village Market Demonstration}\label{sec:village-demo}

To demonstrate that the instruments of Sections~\ref{sec:instruments}--\ref{sec:escrow} compose into a functioning economy, we run a village market of six agents and an escrow agent on the program: Alice~(baker), Bob~(farmer), Charlie~(carpenter), Diana~(doctor), Eve~(teacher), and Frank~(fisherman).  The scenario unfolds over a month and exercises seven distinct financial operations:

\begin{enumerate}
\item \textbf{Symmetric mutual credit.} Alice and Bob each mint 15~coins and swap them 1-for-1, establishing mutual liquidity.  Alice and Charlie do the same with 15~coins each, Charlie and Eve with 10, and Eve and Frank with 10.

\item \textbf{Asymmetric credit (loan).}  Diana lends Bob 20~diana-coins in exchange for 24~bob-bonds maturing on day~25---a 20\% interest rate.  Diana similarly lends Frank 15~coins for 18~bonds maturing day~28.

\item \textbf{Pay.}  Bob pays Alice 5~alice-coins for bread; Eve pays Charlie 6~charlie-coins for a bookshelf; Frank pays Diana 3~diana-coins for a medical checkup; Eve pays Alice 3~alice-coins for bread.

\item \textbf{Portfolio swap.}  Eve trades 5~frank-coins to Charlie for 5~alice-coins, acquiring the currency she needs to buy bread from Alice.

\item \textbf{Escrow.}  Charlie hires Frank to build a dock and deposits 5 frank-coins with the escrow agent for Frank, to be released when the escrow agent's date reaches day~30; the escrow agent's date reaches day~30, and it releases the coins to Frank as payment for the dock.

\item \textbf{Redeem.}  Frank redeems 5~diana-coins from Diana, taking back five of his own bonds she holds---set-off, reducing his exposure to Diana's currency.

\item \textbf{Sale of debt.}  Alice sells 10~bob-coins to Eve for 4~frank-coins---a discounted sale reflecting Alice's desire to reduce exposure to Bob's debt.
\end{enumerate}

Each party runs the compiled program of Appendix~\ref{app:bond-agent} with its mediator, and the seven run concurrently in one process, a scripted person answering each agent's cards in place of the interface, each act once the one before it has taken effect, and the escrow program answering the escrow agent's; each party's date advances by its own calendar, as the operations need it.  Shown once, on Frank's bonds, is what dates add: before Frank's date reaches day~28, Diana presents one of the frank-bonds she holds for a frank-coin, and Frank's agent returns it, as it is not mature by Frank's date; once Frank's date has reached day~28, Diana presents a frank-bond again and takes a frank-coin for it, and her agent records the date Frank warranted, day~28.  The run performs the seven operations and the maturing, and ends with the holdings of Table~\ref{table:village-holdings}.  They reconcile with Conservation of Money (Proposition~\ref{prop:conservation}): the coins and bonds of each issuer held across the seven parties are those they minted---30 alice-coins; 15 bob-coins and 24 bob-bonds; 25 charlie-coins; 35 diana-coins; 20 eve-coins; and 10 frank-coins and 18 frank-bonds, one of which the maturing replaced by a frank-coin---and are those the seven operations and the maturing produce.

\begin{table}[t]
\centering
\small
\begin{tabular}{@{}lp{0.78\textwidth}@{}}
\hline
\textbf{Party} & \textbf{Holdings after the scenario} \\
\hline
Alice & 5 bob-coins, 8 alice-coins, 15 charlie-coins, 4 frank-coins \\
Bob & 10 alice-coins, 20 diana-coins \\
Charlie & 10 alice-coins, 10 eve-coins, 6 charlie-coins \\
Diana & 24 bob-bonds (day 25), 12 frank-bonds (day 28), 8 diana-coins, 1 frank-coin \\
Eve & 4 charlie-coins, 1 frank-coin, 2 alice-coins, 10 bob-coins \\
Frank & 7 diana-coins, 10 eve-coins, 5 frank-bonds (day 28), 5 frank-coins \\
Escrow agent & nothing \\
\hline
\end{tabular}
\caption{The holdings of the seven parties at the end of the village market, as each party's screen reports them, self-held coins and bonds included; at the end Bob's date is 25 and Frank's 28, so every bond listed is mature.}
\label{table:village-holdings}
\end{table}

\subsection{Implementation Summary}\label{sec:impl-summary}

The village market of Section~\ref{sec:village-demo} is a GLP script that drives the seven parties through the seven operations.  The implementation is open-source: the bond agent, its compilation, the escrow programs and the plays are at \repo, in \texttt{programs/currencies/bonds}.

The proof-of-concept represents each coin and each bond as a unit; a practical deployment would use range-denominated coins and bonds to compress holdings of the same issuer and maturity into a single term.  Each agent keeps its own local date $d_p^*$, advanced via Advance-date, with no global clock.  A maturing requires the bond to be mature by the calendar of its issuer (Section~\ref{sec:contract}); the issuer applies their own date, so no party needs to see another's, and the bearer keeps the record of the date the issuer warranted (clause~2).  The scenario runs with every party acting as the agreements specify.  Four things are left to further work: an issuer who defaults by refusing a presentation they are obliged to fulfil, and the effect on the liquidity ratios of the holders of their bonds; counterparties who deviate from an escrow program; concurrent redemptions contending for the same coins and bonds, where the menu of Appendix~\ref{app:bond-agent} is contested; and populations of hundreds of agents, at the scale of the village of Appendix~\ref{sec:community-bank}.

\section{Conclusion}\label{sec:conclusion}

Grassroots bonds add to grassroots coins a second instrument, a debt of its issuer with a maturity date, and with it the interest that 1-for-1 coin redemption precluded.  A coin is money and a bond is credit, and the bearer of a matured bond takes a coin of its issuer for it; the voluntary swap of coins and bonds and the escrow transfers are the acts from which the instruments of Sections~\ref{sec:instruments} and~\ref{sec:escrow} are composed.  The contract of Section~\ref{sec:contract} is the contract of grassroots currencies with two clauses, Date and Escrow; the Date clause is attested, so the contract enforces the maturing of a bond only against the date its issuer warrants, and the obligation of a bond as a whole is attested by the record, where a coin's is enforced.  The ratios of Section~\ref{sec:liquidity} are the classical ones, with maturity as asset class, computed by each person on their own holdings; the clearing results are properties of the holdings, and hold once dates advance and persons act on their rights, with no coordination among them.

The contract has two realisations.  With pen and paper the obligations are written on the notes, the date is that of the calendar, and both are enforced at law.  Digitally they are enforced by the program every participant runs, in which the issuer applies their own date to their own bonds, and no party needs to see another's date.

Further work is the economy that is not a demonstration: issuers who default, counterparties who deviate from an agreement, and populations of thousands rather than six.  Securing the implementation against the loss of a signing key or a device is the subject of prior work~\cite{eitan2026securing}.

\mypara{AI Disclosure}
The authors used Anthropic's Claude models to assist with two parts of this work.  First, the GLP code (Section~\ref{sec:implementation}) was derived from the specification of Sections~\ref{sec:bonds} and~\ref{sec:contract} with substantial Claude assistance: the authors specified the transactions, supplied GLP language documentation, iteratively reviewed the generated code against the specification and against a running implementation, and made all design decisions.  Second, the prose of the paper received Claude assistance for drafting and proofreading under the same review discipline.  The authors verified the correctness and originality of all content, including all references.  Responsibility for all content rests with the authors.

\bibliography{bib}

\newpage
\appendix
\section{The Act Schemas of the Contract}\label{app:schemas}

This appendix writes the clauses of Section~\ref{sec:clauses} as act schemas~\cite{shapiro2026formalising} and shows the contract syntactically grassroots.  Act schemas are written over atoms whose arguments are people, speech acts, or numbers; a schema may carry a constraint on its numeric arguments, written after its head, which every binding satisfies; and an atom added or deleted may carry a multiplicity.  The contract of grassroots currencies has four schemas---mint, swap, pay and redeem---over one predicate, \textcent\ of arity one, a coin naming its issuer~\cite{shapiro2024gc}.  Here \textcent\ gains a maturity date, which a coin takes as $\bot$, redeem takes a coin or a bond, and mature, advance and the three escrow transfers are new schemas.

\subsection{The Act Schemas}\label{sec:gc-schemas}

Three predicates.  \textcent, of arity two: \textcent$(u,d)$ in the state of $p$ is a bond issued by $u$ with maturity date $d$, held by $p$, and \textcent$(u,\bot)$ is a coin of $u$, due now; $d$ ranges over $\calN \cup \{\bot\}$.  $\mathit{date}$, of arity one: $\mathit{date}(t)$ in the state of $p$ is the date of $p$.  $\mathit{dated}$, of arity two: $\mathit{dated}(u,t)$ in the state of $p$ is $u$'s warranty to $p$ that its date was $t$.  The roles are $p$, $q$ and $e$; the party variables are $u$, $v$ and $r$; the numeric variables are $d$, $d'$ for maturity dates, $t$, $t'$ for dates, and $k$, $k'$ for amounts.  The six clauses that describe acts become nine schemas, redemption taking two and escrow three.

\begin{center}
\footnotesize
\setlength{\tabcolsep}{3pt}
\begin{tabular}{lll}
\hline
\emph{schema} & \emph{at its first role} & \emph{at its second role} \\
\hline
$\mathit{mint}(p?;\,d,k,t) \mid d = \bot \lor d > t :$ & $\mathit{date}(t),\ +k\cdot\text{\textcent}(p,d)$ & \\
$\mathit{advance}(p;\,t,t') \mid t' > t :$ & $-\mathit{date}(t),\ +\mathit{date}(t')$ & \\
$\mathit{swap}(p?,q?;\,u,d,k,v,d',k') :$ & $-k\cdot\text{\textcent}(u,d),\ +k'\cdot\text{\textcent}(v,d')$ & $-k'\cdot\text{\textcent}(v,d'),\ +k\cdot\text{\textcent}(u,d)$ \\
$\mathit{pay}(p?,q;\,k) :$ & $-k\cdot\text{\textcent}(q,\bot)$ & $+k\cdot\text{\textcent}(q,\bot)$ \\
$\mathit{redeem}(p?,q;\,r,d') :$ & $-\text{\textcent}(q,\bot),\ +\text{\textcent}(r,d')$ & $-\text{\textcent}(r,d'),\ +\text{\textcent}(q,\bot)$ \\
$\mathit{mature}(p?,q;\,d,t) \mid d \le t :$ & $-\text{\textcent}(q,d),\ +\text{\textcent}(q,\bot),\ +\mathit{dated}(q,t)$ & $\mathit{date}(t)$ \\
$\mathit{deposit}(p?,e?;\,u,d,k) :$ & $-k\cdot\text{\textcent}(u,d)$ & $+k\cdot\text{\textcent}(u,d)$ \\
$\mathit{release}(e?,q?;\,u,d,k) :$ & $-k\cdot\text{\textcent}(u,d)$ & $+k\cdot\text{\textcent}(u,d)$ \\
$\mathit{return}(e?,p?;\,u,d,k) :$ & $-k\cdot\text{\textcent}(u,d)$ & $+k\cdot\text{\textcent}(u,d)$ \\
\hline
\end{tabular}
\end{center}

\noindent ``At its own election'' is the single guarding role of mint, of pay, of redeem and of mature; ``by an act of both'' is the two guarding roles of the swap, and of each escrow transfer; advance is guarded by no role, as the calendar asks no one.  ``Of any amount'' is the argument $k$ of mint, and ``of a maturity later than its own date'' is its constraint: a coin takes $d = \bot$ and a bond a date past the minter's own, so no bond is born mature.  A bond is mature when its maturity date is at most its issuer's date: mature requires $\mathit{date}(t)$ of the issuer $q$, in $q$'s role, under the constraint $d \le t$, so maturity is judged by the issuer's date and the holder's date plays no part, and it adds $\mathit{dated}(q,t)$ at $p$, the issuer's warranty of its date, which is clause~2.  Mature is the one schema that reads a date of another role and the one that adds a warranty; pay and redeem carry coins, $\bot$ in their maturity, and read no date.  ``Any coin or bond that party holds'' is redeem's party variable $r$ and maturity $d'$, unconstrained.  The mutual credit line of clause~3 is the swap at a binding sending $u$ to the party at $p$, $v$ to the party at $q$, $d$ and $d'$ to $\bot$, and $k'$ to $k$; the loan is the swap at a binding with $k' > k$, $d = \bot$, and $d'$ later than the borrower's date.  Release and return differ in who receives: the party at $q$ in release, the depositor at $p$ in return; which of the two the agent wills is fixed by the escrow agreement (Section~\ref{sec:escrow}).  The swap exchanges one lot---coins of one issuer, or bonds of one issuer and maturity---against another; an instrument that exchanges several lots at once is an escrow agreement, the packaged exchange of Section~\ref{sec:escrow}.

\subsection{The Contract Is Grassroots}\label{sec:gc-grassroots}

A contract is syntactically grassroots if it has an unobstructed introductory act, every predicate its acts require has traceable provenance, and every act not willed by all its parties runs only between parties who already record each other; the protocol realising a syntactically grassroots contract is volitionally grassroots~\cite{shapiro2026formalising}.  The three conditions quantify over the people a state records, and the numeric arguments and constraints of the schemas above name no person, so the conditions and the theorem apply to them as stated.  The check for the four schemas of grassroots currencies is in~\cite{shapiro2024gc}; a maturity date names no person, so that check stands, and what follows is the check for the whole contract, the new schemas and predicates included.

\mypara{Introduction}  The swap at the binding of the mutual credit line is an introductory act: it is guarded in both roles, it adds \textcent$(q,\bot)$ at $p$, naming $q$, and \textcent$(p,\bot)$ at $q$, naming $p$.  It is unobstructed: what each role requires, $k$ coins of their own issue, is added by mint, a schema of arity one, guarded in its role, whose one other requirement is $\mathit{date}(t)$, which names no person; and no role forbids anything.

\mypara{Provenance}  \textcent\ has traceable provenance.  Mint adds \textcent$(p,d)$ at $p$, naming its own role.  Swap adds \textcent$(v,d')$ at $p$ and \textcent$(u,d)$ at $q$, and $v$ and $u$ are arguments of the atoms the other role requires.  Pay adds \textcent$(q,\bot)$ at $q$, naming a role.  Redeem adds \textcent$(r,d')$ at $p$, where $r$ is an argument of the atom $q$ requires, and \textcent$(q,\bot)$ at $q$, naming a role.  Mature adds \textcent$(q,\bot)$ at $p$, naming a role.  Each escrow transfer adds at its receiving role an atom naming $u$, an argument of the atom its giving role requires.  $\mathit{dated}$ has traceable provenance: mature adds $\mathit{dated}(q,t)$ at $p$, naming a role.  $\mathit{date}$ names no person.

\mypara{Volition}  Of the schemas of arity two, swap, deposit, release and return are guarded in both roles.  Pay, redeem and mature are guarded by $p$, and in each $p$ requires an instrument of $q$---\textcent$(q,\bot)$ in pay and redeem, \textcent$(q,d)$ in mature---an atom of traceable provenance naming $q$, so the role graph is connected: a party pays, presents or matures with an issuer whose coin or bond they already hold, which is to say already record.

\medskip
\noindent The contract is therefore syntactically grassroots, and the protocol realising it is volitionally grassroots (Corollary~\ref{cor:bonds-grassroots}): any group of persons operates it on its own, and two groups become connected only by a swap willed by a member of each.

\section{Deferred Proofs}\label{app:proofs}

This appendix proves the correctness of the credit line of Section~\ref{sec:escrow} and the clearing results of Section~\ref{sec:liquidity}.  The syntactic conditions establishing that grassroots bonds are volitionally grassroots (Corollary~\ref{cor:bonds-grassroots}) are checked in Appendix~\ref{app:schemas}.

\smallskip\noindent\textbf{Proposition~\ref{prop:credit-line} (Credit Line) (restated).} \emph{Let $k$ be the limit, $T$ the expiry, and $k_d$ the drawn amount of the credit line above.  In any safe run in which $e$ follows the escrow program and $p$ and $q$ transact only as the agreement specifies, between the operations of the agreement---each packaged exchange or transfer complete---$e$ holds \textcent$^{k-k_d}_p$, $p$ holds \textcent$^{k_d}_{q,T}$ together with the coupon bonds of the draws, and $0 \le k_d \le k$.  At expiry with $k_d = 0$, $p$ holds \textcent$^k_p$ and no $q$-bond maturing at $T$.}

\begin{proof}
By induction on the operations of the agreement, in the order they complete in the run.

\mypara{Base}  At establishment $p$ mints \textcent$^k_p$ and transfers them to $e$, and $k_d = 0$.  So $e$ holds \textcent$^{k}_p$, $p$ holds no $q$-bond maturing at $T$, and $0 \le k_d \le k$.

\mypara{Draw}  Let the invariant hold with drawn amount $k_d$, and let $q$ draw $k' \le k - k_d$.  By the invariant $e$ holds \textcent$^{k-k_d}_p \supseteq \text{\textcent}^{k'}_p$, so once $q$ has deposited \textcent$^{k'}_{q,T}$ and the coupon bonds with $e$, the releases complete: they move \textcent$^{k'}_{q,T}$, with the coupon bonds, from $e$ to $p$, and \textcent$^{k'}_p$ from $e$ to $q$.  Afterwards $e$ holds \textcent$^{k-k_d-k'}_p$ and $p$ holds \textcent$^{k_d+k'}_{q,T}$, and the new drawn amount is $k_d + k' \le k$.  The invariant holds.

\mypara{Repayment}  Let the invariant hold with drawn amount $k_d$, and let $q$ repay $k_{\text{rep}} \le k_d$.  By the invariant $p$ holds \textcent$^{k_d}_{q,T} \supseteq \text{\textcent}^{k_{\text{rep}}}_{q,T}$, so each redemption is enabled.  The $k_{\text{rep}}$ redemptions move \textcent$^{k_{\text{rep}}}_p$ from $q$ to $p$ and \textcent$^{k_{\text{rep}}}_{q,T}$ from $p$ to $q$; $p$ then transfers the coins to $e$.  Afterwards $e$ holds \textcent$^{k-k_d+k_{\text{rep}}}_p$ and $p$ holds \textcent$^{k_d-k_{\text{rep}}}_{q,T}$, and the new drawn amount is $k_d - k_{\text{rep}} \ge 0$.  The invariant holds.

\mypara{Expiry}  Once $d_e^* \ge T$, $e$ transfers \textcent$^{k-k_d}_p$ to $p$.  With $k_d = 0$ this is \textcent$^k_p$, and by the invariant $p$ holds \textcent$^{0}_{q,T}$; as $p$ and $q$ transact only as the agreement specifies, $p$ holds no other $q$-bond maturing at $T$.

No other operation of the agreement changes the $p$-coins held by $e$ or the $q$-bonds maturing at $T$ held by $p$: the coupon bonds mature at the payment dates $d_j$, and by hypothesis the parties take no transaction outside the agreement.  Hence the invariant holds between the operations of the agreement throughout the run.
\end{proof}

\smallskip\noindent\textbf{Proposition~\ref{prop:graph-solvency} (Graph Solvency) (restated).} \emph{The holdings of a sink-closed set $V$ clear if and only if $a_p \ge \ell_p$ for every member $p$ of $V$ other than $s$, and cleared holdings leave each such member holding exactly $a_p - \ell_p$ coins and bonds of the sink.}

\begin{proof}
A clearing redemption among members of $V$ surrenders a coin of a member $q \ne s$ and takes a coin or a bond of $q$'s holdings of an issuer $r \ne q$; if $r \ne s$ then $r \in V$, by sink-closedness of the edge $q \rightarrow r$.  The surrendered coin becomes self-held and the instrument taken moves to the redeemer, so every edge the redemption deletes, creates, or moves is incident to $s$ or has both endpoints in $V$; edges outside $V$ are untouched, and $V$ remains sink-closed along any safe computation of clearing redemptions among its members, Mature transactions of its members and Advance-date transitions.

Write $n_u := a_u - \ell_u$.  A Redeem in which holder $p$ surrenders a $q$-coin and takes an instrument $y$ of issuer $r$ held by $q$ changes the counts as follows.  The surrendered coin moves to its issuer: $a_p$ and $\ell_q$ decrease by 1.  If $r \notin \{p, q\}$, $y$ moves from holder $q$ to holder $p$: $a_q$ decreases and $a_p$ increases by 1.  If $r = p$, $y$ moves to its issuer: $a_q$ and $\ell_p$ decrease by 1.  If $r = q$, $y$ moves from its issuer to $p$: $a_p$ and $\ell_q$ increase by 1.  In every case $n_u$ is unchanged for every person $u$.  A Mature replaces a matured $q$-bond by a $q$-coin at the same holder, so the edge to $q$ stands and no count changes, and Advance-date changes no holdings.  Every $n_u$ is therefore invariant along any safe computation of Redeem, Mature and Advance-date transitions.

Only if: cleared holdings of $V$ have $\ell_p = 0$ for every member $p$ of $V$ other than $s$, so there $n_p = a_p \ge 0$; $n_p$ is invariant, so $a_p \ge \ell_p$ in the given configuration.

If: by induction on $C := \sum_{p \in V \setminus \{s\}} \ell_p$.  If $C = 0$ the holdings of $V$ are cleared.  Otherwise some person $p$ holds an instrument of an issuer $q \in V \setminus \{p, s\}$, and $p$ is a member of $V$, as the edge $p \rightarrow q$ is incident to $s$ or lies in $V$.  If that instrument is a bond, $q$ advances $d_q^*$ past its maturity date and $p$ takes a $q$-coin for it, a Mature transaction, which changes no count.  As $a_q \ge \ell_q \ge 1$, $q$ holds an instrument $y$ of some issuer $r \ne q$, and $p$ redeems the $q$-coin against $y$---a clearing redemption among members of $V$.  By the count changes above, the redemption decreases $C$ by 1 or 2 and leaves every $n_u$ unchanged, so the condition persists and the induction applies.  At most $C$ redemptions clear the holdings of $V$.

In cleared holdings of $V$, every instrument of another issuer that a member $p \ne s$ holds is a coin or a bond of the sink---its issuer is $s$ or in $V$ by sink-closedness, and every issuer in $V$ other than $s$ is cleared---and there are $a_p - \ell_p$ of them, as $n_p$ is invariant and $\ell_p = 0$ there.
\end{proof}

\smallskip\noindent\textbf{Proposition~\ref{prop:clearing-liveness} (Clearing Liveness) (restated).} \emph{If $a_p \ge \ell_p$ for every member $p$ of $V$ other than $s$, then every infinite safe computation of clearing redemptions among members of $V$, Mature transactions of members of $V$, and Advance-date transitions of members of $V$ that satisfies eventual redemption reaches cleared holdings of $V$, after at most $\sum_{p \in V \setminus \{s\}} \ell_p$ redemptions.}

\begin{proof}
A clearing redemption among members of $V$ decreases $C = \sum_{p \in V \setminus \{s\}} \ell_p$ by 1 or 2 and leaves every $n_u$ unchanged (proof of Proposition~\ref{prop:graph-solvency}), and Advance-date transitions change no holdings, so the computation takes at most $C$ redemptions.  A Mature decreases by one the number of bonds the members of $V$ hold, and a clearing redemption and an Advance-date transition change that number not at all, so the computation takes finitely many maturings.  Suppose the holdings of $V$ after the last redemption and the last maturing are not cleared: some person $p$ holds an instrument of an issuer $q \in V \setminus \{p, s\}$ in every subsequent configuration, and $p$ is a member of $V$, as the edge $p \rightarrow q$ is incident to $s$ or lies in $V$.  If that instrument is a bond, by~(\ia) $d_q^*$ eventually exceeds its maturity date, and the Mature in which $p$ takes a $q$-coin for it is enabled in every configuration of a suffix in which no maturing occurs, contradicting~(\ib).  So it is a $q$-coin; and as $a_q \ge \ell_q \ge 1$ and no transaction of the suffix changes any holding, $q$ holds an instrument $y$ of some issuer $r \ne q$ in every subsequent configuration.  The clearing redemption in which $p$ surrenders the $q$-coin and takes $y$ is then enabled in every configuration of a suffix, with no clearing redemption occurring in the suffix, contradicting~(\ib).
\end{proof}

\section{Escrow-Based Instruments, Continued}\label{app:escrow-more}

The escrow-based instruments deferred from Section~\ref{sec:escrow}: collateral, guarantees, options, insurance, credit default swaps, and letters of credit.  The notation, the escrow agent $e$, and the trust assumptions are those of Section~\ref{sec:escrow}.

\begin{enumerate}
\item \textbf{Collateral:}  To collateralise a loan from lender $p$ to borrower $q$:

\begin{enumerate}
\item \textbf{Deposit.}  $\{(q, x),\; (e, \emptyset)\}$: borrower $q$ transfers bonds $x$ to escrow agent $e$. 
\item \textbf{Release on default.}  $\{(e, x),\; (p, \emptyset)\}$: $e$ transfers collateral to lender $p$. 
\item \textbf{Return on fulfilment.}  $\{(e, x),\; (q, \emptyset)\}$: $e$ returns collateral to borrower $q$. 
\end{enumerate}

\item \textbf{Guarantee:}  Guarantor $g$ covers the obligations of borrower $q$ to lender $p$:

\begin{enumerate}
\item \textbf{Deposit.}  $\{(g, \text{\textcent}^k_g),\; (e, \emptyset)\}$: $g$ deposits $g$-coins with escrow agent $e$.  
\item \textbf{Invocation on default.}  $\{(e, \text{\textcent}^k_g),\; (p, \emptyset)\}$: $e$ releases $g$-coins to lender $p$. 
\item \textbf{Release on fulfilment.}  $\{(e, \text{\textcent}^k_g),\; (g, \emptyset)\}$: $e$ returns $g$-coins to $g$. 
\end{enumerate}

\item \textbf{Option:}  An option grants $p$ the right, but not the obligation, to execute a specified swap with $q$ at or before expiry date $d$, judged by escrow agent local date $d_e^*$, with establishment-window $T_0$ before which both parties must complete their deposits.  The option premium and the underlying assets are held in escrow:

\begin{enumerate}
\item \textbf{Establishment.}  $p$ deposits the premium \textcent$^k_p$ with escrow agent $e$: $\{(p, \text{\textcent}^k_p),\; (e, \emptyset)\}$.  $q$ deposits the underlying bonds $y$ with $e$: $\{(q, y),\; (e, \emptyset)\}$.  If only one party deposits by establishment-window $T_0$, $e$ returns the deposit to its depositor.
\item \textbf{Exercise.}  While $d_e^* \le d$, $p$ may signal exercise; $e$ releases the underlying $y$ to $p$: $\{(e, y),\; (p, \emptyset)\}$, and releases the premium to $q$: $\{(e, \text{\textcent}^k_p),\; (q, \emptyset)\}$.
\item \textbf{Expiry.}  Once $d_e^* > d$, if $p$ has not exercised, $e$ returns the underlying to $q$: $\{(e, y),\; (q, \emptyset)\}$, and transfers the premium to $q$: $\{(e, \text{\textcent}^k_p),\; (q, \emptyset)\}$.  
\end{enumerate}

The above specifies a \emph{call option}: $p$ acquires the right to buy the underlying $y$ from $q$.  A \emph{put option}---the right to sell bonds to $q$ at a predetermined price---is symmetric: $q$ deposits the strike price with $e$, and $p$ deposits the bonds they wish to sell.

Exercise is specified by a window $[d_1, d_2]$ agreed at establishment: $p$ may exercise while $d_1 \le d_e^* \le d_2$, and expiry is once $d_e^* > d_2$.  The specification above uses \emph{American} exercise, the window $[0, d]$; \emph{European} exercise is the window $[d, d_2]$, with $d_2$ the last date on which $e$ accepts exercise.  A window is needed because Advance-date is unguarded and may raise $d_e^*$ by more than one day, so $d_e^*$ can pass $d$ without equalling it and the condition $d_e^* = d$ may never hold.

A \emph{strike price} can be incorporated by having $p$ deposit additional coins \textcent$^{k'}_p$ with $e$ at establishment.  On exercise, $e$ transfers the strike price to $q$ along with the premium; on expiry, $e$ returns the strike price to $p$ and transfers the premium to $q$.

\item \textbf{Insurance:}  Party $p$ (the insured) pays a premium to escrow agent $e$, who disburses a payout if a specified event occurs by contract expiry $T$ (judged by $d_e^*$), as attested by $e$:

\begin{enumerate}
\item \textbf{Establishment.}  $\{(q, \text{\textcent}^{k'}_q),\; (e, \emptyset)\}$: the insurer $q$ deposits the reserve \textcent$^{k'}_q$, $k' \gg k$, with $e$.
\item \textbf{Premium.}  $\{(p, \text{\textcent}^k_p),\; (e, \emptyset)\}$: $p$ deposits the premium with $e$.  
\item \textbf{Claim (event occurs).}  $\{(e, \text{\textcent}^{k'}_q),\; (p, \emptyset)\}$ and $\{(e, \text{\textcent}^k_p),\; (q, \emptyset)\}$: $e$ transfers the reserve to $p$ as the payout, and the premium to the insurer $q$. 
\item \textbf{Expiry (no event).}  $\{(e, \text{\textcent}^k_p \cup \text{\textcent}^{k'}_q),\; (q, \emptyset)\}$: once $d_e^* \ge T$ with no claim, $e$ transfers the premium to insurer $q$ and returns its reserve. 
\end{enumerate}

The escrow agent $e$ serves as the claims adjudicator.  A mutual insurance arrangement among a group of participants can be realised by each member depositing premiums into a shared escrow reserve.

\item  \textbf{Credit default swap (CDS):}  Party $p$ pays periodic premiums to $q$ (the protection seller), and $q$ compensates $p$ if a reference entity $r$ defaults.  Escrow agent $e$ adjudicates:

\begin{enumerate}
\item \textbf{Establishment.}  $\{(q, \text{\textcent}^K_q),\; (e, \emptyset)\}$: the protection seller $q$ deposits the reserve \textcent$^K_q$ with $e$.
\item \textbf{Premium payments.}  Periodic swaps $\{(p, \text{\textcent}^{k_j}_{p,d_j}),\; (e, \emptyset)\}$ at dates $d_j$ judged by $d_e^*$: $p$ deposits premium bonds with $e$, forwarded to $q$. 
\item \textbf{Credit event.}  $\{(e, \text{\textcent}^K_q),\; (p, \emptyset)\}$: if $e$ determines that $r$ has defaulted, $e$ releases the reserve coins of $q$ to $p$. 
\item \textbf{No credit event.}  $\{(e, \text{\textcent}^{k_j}_{p,d_j}),\; (q, \emptyset)\}$ and, at contract expiry $T$ (when $d_e^* \ge T$), $\{(e, \text{\textcent}^K_q),\; (q, \emptyset)\}$: premium bonds are released to $q$, and the reserve coins are returned to $q$.  
\end{enumerate}

\item  \textbf{Letter of credit:}  A bank $b$ guarantees payment to seller $p$ on behalf of buyer $q$, with escrow agent $e$ verifying that contractual conditions (e.g.\ delivery of goods) are met:

\begin{enumerate}
\item \textbf{Issuance.}  $\{(b, \text{\textcent}^k_b),\; (e, \emptyset)\}$: bank $b$ deposits $b$-coins with $e$. 
\item \textbf{Presentation.}  $\{(e, \text{\textcent}^k_b),\; (p, \emptyset)\}$: upon $e$ verifying that $p$ has fulfilled the terms (e.g.\ delivery attestation), $e$ releases $b$-coins to seller $p$.  
\item \textbf{Reimbursement.}  $\{(q, \text{\textcent}^k_{q,d}),\; (e, \emptyset)\}$ and $\{(e, \text{\textcent}^k_{q,d}),\; (b, \emptyset)\}$: buyer $q$ deposits their bonds with $e$, which releases them to bank $b$.  The reimbursement goes through $e$ because $q$-bonds are not coins of $b$ and so cannot be paid to it, and there is no other act by which bonds pass one way between two ordinary parties (Section~\ref{sec:escrow}). 
\end{enumerate}

\end{enumerate}

\section{The Liquidity Ratios: Worked Example}\label{app:liquidity-more}

We compute the cash, quick, and current ratios of Section~\ref{sec:liquidity} for two agents of the village market scenario (Section~\ref{sec:village-demo}), at the snapshot immediately after step~2 (asymmetric credit), with both agents' local date $d^* = 0$, near-term horizon $\delta = 90$ days, and operating cycle $\Delta = 1$ year.  At this point the relevant holdings are as follows.

Bob holds 15~alice-coins (from step~1's symmetric mutual credit with Alice) and 20~diana-coins (from the loan Diana extended in step~2): 35~foreign coins.  Bob has issued 15~bob-coins (held by Alice) and 24~bob-bonds (held by Diana, maturing day~25); the coins are due now and the bonds mature within $\Delta$, so the current liabilities are 39.

Diana holds 24~bob-bonds (maturing day~25) and 18~frank-bonds (maturing day~28), 42~foreign bonds in total, and no foreign coins.  Diana has issued 20~diana-coins (held by Bob) and 15~diana-coins (held by Frank); all 35~diana-coins are due now and hence within $\Delta$, so the current liabilities are 35.

Applying the formulas with common denominator (current liabilities):
$$
\begin{array}{lll}
\textit{Cash Ratio of Bob}    &= 35/39 &\approx 0.90 \\
\textit{Quick Ratio of Bob}   &= 35/39 &\approx 0.90 \\
\textit{Current Ratio of Bob} &= 35/39 &\approx 0.90 \\[4pt]
\textit{Cash Ratio of Diana}    &= 0/35  &= 0 \\
\textit{Quick Ratio of Diana}   &= 42/35 &\approx 1.20 \\
\textit{Current Ratio of Diana} &= 42/35 &\approx 1.20
\end{array}
$$

Bob's three ratios coincide: he holds only foreign coins, so widening the horizon adds nothing to his numerator.  His common value of 0.90 reflects that his 35~foreign coins fall short of his 39~current liabilities.

Diana's cash ratio is zero---she holds no foreign coins on hand---but her quick and current ratios both reach 1.20, reflecting full coverage of her current liabilities by the bob-bonds and frank-bonds she will collect within the operating cycle: Diana is illiquid and solvent (Section~\ref{sec:liquidity}).  Her quick and current ratios coincide because all foreign bonds she holds mature within $\delta$; in a scenario with longer-dated holdings the current ratio would exceed the quick.

\section{The Community Bank with Bonds}\label{sec:community-bank}

A community bank, or credit union, issues its own grassroots currency and opens credit lines with community members, whose members undertake to accept payments in the community currency; it needs no initial capital, and it manages its exposure to a member by redeeming that member's coins against the coins of other members it holds~\cite{shapiro2024gc}.  Bonds add to it lending at interest, deposits, collateral, and risk management by maturity.

\mypara{Bank lending via bonds}  The bank lends community coins to member $p$ in exchange for $p$-bonds with maturity and interest.  For example, the bank mints 10,000 community coins for $p$, and $p$ mints 10,500 $p$-bonds maturing in six months---500 bonds as interest.  The bank holds interest-bearing claims on $p$; $p$ holds liquid community coins accepted by all members.  This is a loan of Sections~\ref{sec:instruments} and~\ref{sec:escrow} (zero-coupon, balloon, or fixed-payment), with the bank as lender, or a credit line (Section~\ref{sec:escrow}), with periodic interest on the amount drawn.  The bank may increase its capital by charging interest on drawn credit, as well as transaction fees.

\mypara{Deposits}  A member may deposit community coins with the bank and receive community bonds maturing at $d$ with interest---for example, 10,000 community coins deposited, 10,200 community bonds maturing in six months returned.  The bank's income is the spread between lending and deposit rates.

\mypara{Collateral}  Personal coins---those acquired by community members through bilateral mutual credit---serve as collateral for bank loans.  The bank may make the loan conditional on borrower $p$ depositing non-$p$-coins they hold into escrow (Appendix~\ref{app:escrow-more}) as security.  If $p$ defaults, the escrowed bonds are released to the bank.  If $p$ fulfils the loan, the collateral is returned.  The bank does not hold or manage personal coins directly; it holds $p$-bonds and, when needed, escrowed collateral.

\mypara{Risk management}  The bank manages risk through the maturity structure of its loan portfolio, collateral requirements at origination, the liquidity ratios of Section~\ref{sec:liquidity} applied to its borrowers, and the ability to sell debt at a discount (Section~\ref{sec:instruments}, sale of debt).  As in traditional banking, a fixed-term loan cannot be recalled before maturity; the bank must arrange protection at origination or sell the debt to a third party.  Risk management by redemption~\cite{shapiro2024gc} reaches only the member's coins: their bonds held by the bank become coins as they mature.

\section{From the Transactions to the Program}\label{app:bond-agent}

The program of Section~\ref{sec:implementation} is written in vGLP~\cite{shapiro2026volition}, GLP with volition guards, and derived from Definition~\ref{def:gb-transactions} as the program of grassroots currencies is derived from its four transactions~\cite{shapiro2024gc}: one volition-guarded clause per transaction, guarded where the transaction is guarded, and one clause of the receiver per act that arrives from another party.  The swap and its responder, pay and redeem, and what the vGLP compiler derives from the source---the interface constructs and the mediator---are as in the coins paper, and are not restated; here is what dates and escrow add, and the channel two parties transact over.  The program is at \repo, in \texttt{programs/currencies/bonds}: the source \texttt{bonds\_agent.vglp}, its compilation \texttt{bonds\_agent.glp}, the escrow programs \texttt{escrow.glp}, and the plays \texttt{village.glp} and \texttt{credit\_line.glp}.

\mypara{The agent}
Each party runs a \emph{bond agent}, the procedure \verb|agent| of the source, whose state is the machine state of Definition~\ref{def:gb-transactions}: its holdings, each coin as \verb|coin(Issuer, Serial)| and each bond as \verb|bond(Issuer, Maturity, Serial)|; its date; the warranties it has received, each as \verb|dated(Issuer, Date)|; and the serial its next instrument takes.  Its arguments are its name, its own input, the network input, the calendar, its outputs, and the four components of its state.  A lot---coins of one issuer, or bonds of one issuer and maturity---is \verb|lot(Issuer, Due, Count)|, the \textcent$^k_{u,\ast}$ of the definition, its due \verb|now| for coins and the maturity for bonds.

\mypara{The channel}
Two parties transact over a channel, and the first act between them carries it: a payment, a presentation or a maturing, which the receiver owes, opens the channel at once; a swap or a transfer opens it when the receiver accepts, and a declined one opens none.

\mypara{Mint}
Mint$_p(k)$ is as in the coins paper.  Mint$_p(d,k)$ is a request clause with two fields, the amount and the maturity, whose guard carries the precondition $d > d^*_p$, so no bond is born mature, and whose effect is $c'_p = c_p \uplus \text{\textcent}^k_{p,d}$:
\begin{verbatim}
*(K, D)
agent(Id, UserIn, NetIn, Cal, Outs, Holdings, Date, Ws, Serial) :-
    ground(Id?), integer(K?), K? >= 1, integer(D?), D? > Date? |
    create_bonds(Id?, D?, K?, Serial?, New),
    Serial1 := Serial? + K?,
    append_instruments(Holdings?, New?, Holdings1),
    show(minted(K?, D?), Outs?, Outs1),
    after_change(Id?, UserIn?, NetIn?, Cal?, Outs1?, Holdings1?,
                 Date?, Ws?, Serial1?).
\end{verbatim}

\mypara{Advance-date and the calendar}
Advance$_p(t)$ is guarded by no one, so it is an ordinary clause, not a volition-guarded one: the calendar reaches the agent as a stream of dates, its fourth argument, and a date past the agent's advances it---the precondition $t > d^*_p$ and the effect $d'^*_p = t$---while one not past it is no transaction and is dropped:
\begin{verbatim}
agent(Id, UserIn, NetIn, [T|Cal], Outs, Holdings, Date, Ws,
      Serial) :-
    ground(Id?), T? > Date? |
    show(date_advanced(T?), Outs?, Outs1),
    agent(Id?, UserIn?, NetIn?, Cal?, Outs1?, Holdings?, T?, Ws?,
          Serial?).

agent(Id, UserIn, NetIn, [T|Cal], Outs, Holdings, Date, Ws,
      Serial) :-
    ground(Id?), T? =< Date? |
    agent(Id?, UserIn?, NetIn?, Cal?, Outs?, Holdings?, Date?,
          Ws?, Serial?).
\end{verbatim}
Who writes the calendar is the platform's affair, not the program's: on a device the clock writes it; in the headless run below the play writes each party's, so that the dates advance as the scenario needs them.  Two parties may therefore disagree on the date, as Section~\ref{sec:bonds} allows, and nothing in the program compares one party's date with another's.

\mypara{Pay and redeem}
Pay$_{p,q}(k)$ and Redeem$_{p,q}(r,\ast)$ read no date.  A payment carries $q$'s coins, which $q$'s agent takes without asking its person, returning anything else.  A presentation carries a $q$-coin and names the issuer and the due of the instrument wanted, \verb|now| for a coin and the maturity for a bond; $q$'s agent takes its coin back and gives that instrument, and where it holds none it returns the coin with the lots of other issuers it holds, the menu of the coins paper with the due of each lot, for the bearer to choose among.

\mypara{Mature}
Mature$_{p,q}(d)$ is guarded by $p$: a request clause with two fields, the issuer and the maturity, in which the precondition $\text{\textcent}_{q,d} \in c_p$ is the selection of the bond:
\begin{verbatim}
*(Q, D)
agent(Id, UserIn, NetIn, Cal, Outs, Holdings, Date, Ws, Serial) :-
    ground(Id?), ground(Q?), integer(D?) |
    take_bond(Q?, D?, Holdings?, Found, Rest),
    do_mature(Found?, Id?, Q?, D?, Rest?, UserIn?, NetIn?, Cal?,
              Outs?, Date?, Ws?, Serial?).
\end{verbatim}
The bond reaches $q$ with a response variable, and $q$'s agent answers without asking its person, as clause~5 obliges.  A bond of $q$ mature by $q$'s date---the precondition $d \le d^*_q$---is answered with a fresh $q$-coin and $q$'s date, the warranty of clause~2; any other bond goes back:
\begin{verbatim}
grant(bond(I, M, _), Id, Date, Serial,
      mature_give(coin(Id?, Serial?),
      Date?), From, UserIn, NetIn, Cal, Outs, Holdings, Ws) :-
    I? =?= Id?, M? =< Date? |
    Serial1 := Serial? + 1,
    show(matured_for(From?, M?), Outs?, Outs1),
    agent(Id?, UserIn?, NetIn?, Cal?, Outs1?, Holdings?, Date?,
          Ws?, Serial1?).
grant(bond(I, M, S), Id, Date, Serial,
      mature_declined(bond(I?, M?, S?)), From, UserIn, NetIn, Cal,
      Outs, Holdings, Ws) :-
    otherwise |
    show(not_mature_for(From?, M?), Outs?, Outs1),
    agent(Id?, UserIn?, NetIn?, Cal?, Outs1?, Holdings?, Date?,
          Ws?, Serial?).
\end{verbatim}
The coin replaces the bond at the bearer, so the holdings of the issuer do not change.  The date the issuer answers with is recorded at $p$, the set union $w'_p = w_p \cup \{(q, d^*_q)\}$ of the definition, by the clause that receives the coin:
\begin{verbatim}
agent(Id, [matured_by(Q, D, Coin, T)|UserIn], NetIn, Cal, Outs,
      Holdings, Date, Ws, Serial) :-
    ground(Id?), ground(Q?), integer(D?), ground(Coin?),
        integer(T?) |
    append_instruments(Holdings?, [Coin?], Holdings1),
    add_warranty(Q?, T?, Ws?, Ws1),
    show(matured(Q?, D?), Outs?, Outs1),
    show(warranty(Q?, T?), Outs1?, Outs2),
    after_change(Id?, UserIn?, NetIn?, Cal?, Outs2?, Holdings1?,
                 Date?, Ws1?, Serial?).
\end{verbatim}
The warranties are the record of clause~2 (Section~\ref{sec:modalities}): a bond returned records nothing, as no transaction took place.

\mypara{The escrow transfers}
Deposit$_{p,e}$, Release$_{e,q}$ and Return$_{e,p}$ are transfers of one lot guarded by both participants, so each is a request clause at the giver and a responder at the receiver, as the swap is.  The three request clauses differ in the kind of transfer they name, which the receiver's card shows; the deposit is
\begin{verbatim}
*(E, U, D, K)
agent(Id, UserIn, NetIn, Cal, Outs, Holdings, Date, Ws, Serial) :-
    ground(Id?), ground(E?), ground(U?), ground(D?), integer(K?),
        K? >= 1 |
    select_instruments(U?, D?, K?, Holdings?, Status, Give, Rest),
    propose_transfer(Status?, deposit, Id?, E?, lot(U?, D?, K?),
                     Give?, Rest?, UserIn?, NetIn?, Cal?, Outs?,
                     Date?, Ws?, Serial?).
\end{verbatim}
and the responder's accept clause, with the else-branch that declines an unanswered transfer on the deadline, is
\begin{verbatim}
*(yes, Kind?, From?, Given?)
respond_transfer(offer_transfer(Kind, From, Given), Link, Resp?,
                 Ans?) :-
    ground(Kind?), ground(From?), ground(Given?) |
    accept_transfer_answer(Kind?, From?, Given?, Resp, Link?, Ans)
 *(no) decline_transfer_answer(Kind?, From?, Given?, Resp, Link?,
                               Ans).
\end{verbatim}
The instruments given are set aside at the giver until the receiver answers, and go to the receiver on acceptance or back on refusal; the transfer is taken only when both guards are willing, as the definition states.

\mypara{The escrow program}
The escrow agent $e$ of Section~\ref{sec:escrow} is an ordinary party running the same program.  What makes it an escrow agent is the \emph{escrow program} of \texttt{escrow.glp}, which takes the place of its person: at every transfer card it answers whether $e$ accepts, and at every message of $e$'s screen it says which releases and returns $e$ wills, which are then performed on $e$'s release and return forms as a person would perform them.  For a date-based condition it wills by $e$'s date.  The village market's escrow agreement names the depositor, the beneficiary, the lot and the release date; the deposit of that lot by that depositor is accepted, and once the lot is held and $e$'s date has reached the release date, the release is willed:
\begin{verbatim}
dated_check(Date, T, Dep, Ben, lot(U, D, K), deposited,
            [release(Ben?, U?, D?, K?)],
            dated_escrow(Dep?, Ben?, lot(U?, D?, K?), T?,
                         released, Date?)) :-
    Date? >= T? | true.
dated_check(Date, T, Dep, Ben, Lot, Ph, [],
            dated_escrow(Dep?, Ben?, Lot?, T?, Ph?, Date?)) :-
    otherwise | true.
\end{verbatim}
The credit line of Section~\ref{sec:escrow} is the second escrow program of the file, and the one in which $e$ keeps state across the transactions, the drawn amount: the lender's deposit of the limit opens the line; the borrower's deposits of a principal lot maturing at the expiry, within the undrawn amount, and of its coupon lot are a draw, on which $e$ wills the release of the principal and the coupon bonds to the lender and of that many of the lender's coins to the borrower; the lender's later deposits of their coins, redeemed from them by the borrower against the principal, are repayments, which restore the capacity; and $e$'s date reaching the expiry is the expiry, on which $e$ wills the return of the undrawn coins.  A deposit the program cannot yet accept stays open and is put to it again at every change of its state, so the order in which the parties reach the escrow agent does not matter.

\mypara{The village market}
The play \texttt{village.glp} runs the six villagers and the escrow agent of Section~\ref{sec:village-demo} in one process: each with its mediator, its network input routed from the others' first contacts, its calendar, and a scripted person in place of the interface, which reads the cards the mediator shows and the screen the agent writes, and answers the cards in the order of its script, each act once the one before it has taken effect on its screen; at the escrow agent the escrow program answers the transfer cards and the person performs the release it wills.  The play drives the dates: the escrow agent's reaches day~30 once the deposit is held, Bob's reaches~25 and Frank's~28 by their own calendars.  Shown once is what dates add: before Frank's date reaches~28, Diana presents one of his bonds maturing on day~28 for a frank-coin, and Frank's agent returns it, as it is not mature by his date; once his date has reached~28, she presents one again and takes a frank-coin for it, and her agent records the date he warranted.  The screen messages are logged, and the final holdings of Section~\ref{sec:village-demo} are read off the log.  The play \texttt{credit\_line.glp} runs the credit line on the escrow program---establishment, a draw, repayment in two parts, and expiry---and the screens show what Proposition~\ref{prop:credit-line} states: after the establishment, the draw and each repayment the escrow agent holds the limit less the amount drawn; after the draw the lender holds the principal and the coupon bonds; and at expiry, with nothing drawn, the lender holds the limit in her own coins and no principal bond.

\section{Related Work}\label{sec:related-work}

The relation of grassroots currencies to global cryptocurrencies and the attacks each design is open to, to credit networks, IOU networks and mutual credit systems, to personal currencies, to microfinance and informal community finance, to payment and state channels, and to monetary theory is given in~\cite{shapiro2024gc}.  Here we relate grassroots bonds and the instruments built on them to decentralised finance, to financial contract specification, to tokenised and traditional bonds, and to bills of exchange.

\mypara{Financial instruments on blockchains}
The programmability of blockchains via smart contracts~\cite{szabo1997formalizing,buterin2014next,scilla2018} has given rise to Decentralised Finance (DeFi)~\cite{werner2022sok,schar2021defi}, which re-implements traditional financial instruments on-chain.  Lending protocols~\cite{gudgeon2020loanable,leshner2020compound} offer overcollateralised loans with algorithmically-determined interest rates; decentralised exchanges provide token trading via automated market makers; and stablecoins offer price-stable media of exchange.

These DeFi instruments have analogues among the grassroots bond instruments of Sections~\ref{sec:instruments}--\ref{sec:escrow}: loans, credit lines, sale of debt, and options are all expressible as voluntary swaps of grassroots bonds and escrow agreements, without requiring a blockchain or consensus.  Two differences are noteworthy.  First, DeFi instruments require overcollateralisation enforced by smart contracts because counterparty identity is pseudonymous and trust is absent; grassroots instruments operate among persons with mutual trust, and collateral (Appendix~\ref{app:escrow-more}) is available but not mandatory.  Second, DeFi instruments incur gas fees for every state transition; grassroots instruments incur no fees beyond bilateral communication.  Flash loans~\cite{qin2021attacking}, which exist within a single atomic blockchain transaction, have no counterpart: grassroots bonds have no global atomic transaction spanning many parties, and a loan is a swap of two, settled by redemption at maturity.

\mypara{Financial contract specification}
The foundational work on composable financial contract specification predates blockchain: Peyton Jones et al.~\cite{peytonjones2000composing} introduced a combinator library for describing financial contracts declaratively, inspiring subsequent blockchain domain-specific languages.  Marlowe~\cite{seijas2018marlowe}, Cardano's financial contract DSL, specifies instruments declaratively with exhaustive static analysis; Scilla~\cite{scilla2018} introduced automata-based contracts with formal safety guarantees.  These languages demonstrate that financial instruments can be specified declaratively, yet they compile to consensus-bound smart contracts.  Grassroots bonds take the same declarative approach---each instrument is composed of the acts of a grassroots social contract (Section~\ref{sec:contract}), a swap of one lot, or a packaged exchange of several held by an escrow agent---but enforce it bilaterally rather than via consensus, with each issuer guaranteeing the security of transactions in their currency~\cite{lewis2023grassroots}.  Discreet Log Contracts~\cite{dryja2017discreet} are bilateral derivatives on Bitcoin whose payout is fixed by an oracle signature, with the terms of the contract invisible on chain; the escrow instruments of Section~\ref{sec:escrow} share the bilateral form and the oracle, and settle with no chain at all---the payout is a transfer of bonds by the escrow agent, enforced by the written obligations of the bonds rather than by a spending transaction.

\mypara{Tokenised bonds on blockchain}
On-chain bond issuance represents the blockchain-native approach to the same problem grassroots bonds address.  A Federal Reserve analysis~\cite{carapella2023tokenization} examines tokenisation of real-world assets including the Santander (2019) and EIB (2021) bond issuances on Ethereum, identifying smart contract designs and financial stability implications. Security token offerings (STOs) extend this trend to debt instruments more broadly~\cite{lambert2022stos}.  These tokenised bonds and STOs differ from grassroots bonds in three respects: they require a public blockchain and its associated gas costs; they rely on custodians to bridge on-chain tokens with off-chain legal obligations; and they are issued by institutions rather than by any person.

\mypara{Traditional bond instruments}
The loan structures of Sections~\ref{sec:instruments} and~\ref{sec:escrow}---zero-coupon, balloon, and fixed-payment loans---correspond to standard fixed-income instruments as classified by Fabozzi~\cite{fabozzi2021bonds}.  The escrow-based instruments of Section~\ref{sec:escrow} and Appendix~\ref{app:escrow-more}---collateral, guarantees, and options---likewise have standard counterparts.  Grassroots bonds differ from traditional bonds not in their financial structure but in three respects: they are issued by any person rather than by institutions, traded bilaterally rather than on exchanges, and enforced bilaterally, on the record of the issuer's own date, rather than solely by legal contract.  The economic effect is to democratise access to the bond market: any person can issue, hold, and trade bonds, with the community bank of Appendix~\ref{sec:community-bank} illustrating how a village can achieve the financial functionality of a bank using only grassroots bonds.

\mypara{Bills of exchange and promissory notes}
Before banks dominated payment, debts circulated as bills of exchange and promissory notes---signed promises to pay a fixed sum of a currency at a set date---traded and discounted before maturity~\cite{rogers1995early,santarosa2015financing}.  A grassroots bond is such an instrument in digital form: a signed, dated, transferable debt, sold before maturity at a discount (sale of debt, Section~\ref{sec:instruments}), with the protocol serving the purpose that transfer and proof-of-holdership rules serve for paper notes~\cite{shapiro2024gc}; and each bond is the debt of its issuer: past holders bear no liability for it.  A bill of exchange falls due by the calendar; a grassroots bond falls due by the calendar of its issuer, who warrants it (Section~\ref{sec:contract}).

\end{document}